\documentclass[prx,twocolumn,superscriptaddress,showpacs,nofootinbib,longbibliography]{revtex4-2}

\usepackage[dvips]{graphicx}
\usepackage{subcaption}
\usepackage{amssymb}
\usepackage{float}
\usepackage{amsmath}
\usepackage{natbib}
\usepackage{xcolor}
\usepackage{chngcntr}
\usepackage{enumitem}
\usepackage{hyperref}
\usepackage[normalem]{ulem}
\usepackage{amsthm}
\newtheorem{theorem}{Theorem}

\newtheorem{definition}{Definition}[section]

\hypersetup{
colorlinks=false,
linktocpage=true,
citebordercolor=[rgb]{0.5,0.5,1.0},
linkbordercolor=[rgb]{1.0,0.5,0.5},
urlbordercolor=[rgb]{0.5,0.5,1.0}
}

\begin{document}

\title{Hamiltonian Cycles on Ammann-Beenker Tilings}

\author{Shobhna Singh} 
\email{shobhna862@gmail.com}
  \affiliation{School of Physics and Astronomy, Cardiff University, The Parade, Cardiff CF24 3AA, United Kingdom}
\author{Jerome Lloyd}
  \affiliation{Department of Theoretical Physics, University of Geneva, 24 rue du G{\'e}n{\'e}ral-Dufour, 1211 Gen{\`e}ve 4, Switzerland}
\author{Felix Flicker}
  \affiliation{School of Physics, Tyndall Avenue, Bristol, BS8 1TL, United Kingdom}
  
\begin{abstract}
    We provide a simple algorithm for constructing Hamiltonian graph cycles (visiting every vertex exactly once) on a set of arbitrarily large finite subgraphs of aperiodic two-dimensional Ammann-Beenker (AB) tilings. Using this result, and the discrete scale symmetry of AB tilings, we find exact solutions to a range of other problems which lie in the complexity class NP-complete for general graphs. These include the equal-weight traveling salesperson problem, providing, for example the most efficient route a scanning tunneling microscope tip could take to image the atoms of physical quasicrystals with AB symmetries; the longest path problem, whose solution demonstrates that collections of flexible molecules of any length can adsorb onto AB quasicrystal surfaces at density one, with possible applications to catalysis; and the three-coloring problem, giving ground states for the $q$-state Potts model ($q\ge 3$) of magnetic interactions defined on the planar dual to AB, which may provide useful models for protein folding.
\end{abstract}

\maketitle 

%
\section{Introduction}
\label{sec:Intro}
%

Fractal geometric loop structures are ubiquitous in physics, appearing often as the natural degrees of freedom in models of critical phenomena. Examples include the fluctuating domain walls of an Ising magnet, the compact shapes of long polymer chains~\cite{chan1989compact}, or the space-time trajectories of quantum particles. A special class of loops, attracting much attention beyond physics, are the Hamiltonian cycles. A Hamiltonian cycle of a graph is a closed, self-avoiding loop that visits every vertex precisely once. Study of these objects dates at least to the ninth century A. D., when the Indian poet Rudra\d{t}a constructed a poem based on a `Knight's tour' of the chessboard. Since then they have appeared in a variety of applications in the sciences and mathematics including protein folding~\cite{flory1956statistical,Li1996}, traffic models, spin models in statistical mechanics~\cite{stanley1968dependence,blote1994fully}, and ice-type models of geometrically frustrated magnetism~\cite{baxter,jaubert2011analysis}.

Hamiltonian cycles have been used to model the statistics of polymer melts~\cite{cloizeaux1991polymers,chan1989compact,bodroza2013enumeration,jannink1990polymers}: For example, the critical scaling exponents were calculated for polymer chains in 2D using the self-avoiding walk on the honeycomb lattice~\cite{nienhuis1982exact,jacobsen2004conformal}. They have also played a key role in the study of protein folding~\cite{flory1956statistical,creighton1990protein,DillEA95,chan1989intrachain,chan1989compact,lau1989lattice}. The shapes (``conformational properties") of proteins play a significant role in their biological function. Proteins often adopt remarkably compact and symmetrical structures compared to the more general class of polymers~\cite{DillEA95}. A key question is how to predict 3D conformational properties from the 1D sequences of amino acids from which the proteins are built. One route is to study``simple exact models" in which the amino acids are represented by structureless units, each occupying the vertex of a graph. Neighboring units in the protein must be nearest neighbors on the graph. The graphs are chosen to have specific geometrical embeddings, typically periodic lattices. While not providing accurate geometrical models, such approaches have the advantage of being exactly solvable. They are also able to capture `nonlocal' interactions, in which amino acids interact when they are neighbors on the graph but not on the 1D chain itself~\cite{creighton1990protein, jacobsen2007exact, chan1989compact, agarwala1997local, bodroza2013enumeration}. Remarkably, 2D graphs often give qualitatively similar results to 3D graphs while affording a significant computational advantage~\cite{DillEA95}. Simple exact models have also been instrumental in the development of polymer physics more generally~\cite{orr1947statistical,barber1970random,de1979scaling,freed1987renormalization,cloizeaux1991polymers}.

The $O(n)$ model in statistical physics, which describes $n$-component spins interacting with their nearest lattice neighbors with isotropic couplings~\cite{stanley1968dependence,blote1994fully}, can be mapped to a problem involving self-avoiding loops~\cite{jacobsen2007exact, jacobsen1999universality}. The probability of a given loop configuration is weighted according to two fugacities: $x$ weighting the total perimeter of the loops, and $n$ the number of loops. The $O(n)$ phase diagram encapsulates many well-known models in statistical physics. For example, the ferromagnetic Ising model arises for $n = 1$, where the loops represent the domain walls in the low temperature limit (thus the Ising spins are defined on the planar dual lattice), while in the high temperature limit the natural degrees of freedom are again loops, though their expression in terms of the original spin variables is less intuitive. the $q$-state Potts models can be related to the $n = \sqrt{q}$ limit. as $x\rightarrow\infty$ every lattice site is visited by a loop and one arrives at the so-called ``fully-packed loop" (FPL) models, which additionally model crystal surface growth. and in the $x\rightarrow \infty, n\rightarrow 0$ limit, the system is forced into a single macroscopic loop traversing every site of the system, which is a Hamiltonian cycle~\cite{jacobsen1999universality, kondev1997liouville}. Studying Hamiltonian cycles can thus lead to an understanding of a variety of other models in a given system, and for this reason a great deal of work has been put into their study -- primarily in the simplest context of periodic regular lattices~\cite{batchelor1994exact, lieb1967exact, suzuki1988evaluation, batchelor1994exact, higuchi1998field, kondev1996operator, batchelor1996critical, nienhuis1982exact, bodroza2013enumeration,jacobsen2007exact,bodrovza2016some,kondev1997liouville, blote1994fully, kast1996correlation, orland1985evaluation, jacobsen1998field, kondev1998conformational, schmalz1984compact, agarwala1997local, peto2007generation}. Given, however, that the complex fractal structures  (e.g.~ polymer-protein structures) these objects aim to model often lack translational symmetry and favor disordered growth~\cite{melkikh2018generalized}, studies of Hamiltonian cycles in settings where translational symmetry is absent may unearth important clues toward the universality of these results.  

In this paper we present a simple algorithm for constructing Hamiltonian cycles and fully packed loops on a set of arbitrarily large finite subgraphs of infinite graphs which do not admit periodic planar embeddings. Specifically, we consider graphs formed from subsets of the edges and vertices of Ammann-Beenker (AB) tilings (Fig.~\ref{fig:AB_H_cycle}). These are two-dimensional (2D) infinite aperiodic tilings~\cite{lloyd2021statistical,GrunbaumShephard, ammann1992aperiodic} built from two tiles, a square and a rhombus. They lack the translational symmetry of periodic lattices but nevertheless feature long-range order, evidenced by sharp spots in their diffraction patterns~\cite{baake2012mathematical,GrunbaumShephard,Senechal}. The long-range order of AB enables exact analytical results to be proven, while their infinite extent allows consideration of the thermodynamic limit of an infinite number of vertices, of interest in condensed matter physics where physical quasicrystals are known with the symmetries of AB~\cite{wang1987two,Wang88}. Recent experiments have demonstrated tunable quasicrystal geometries in twisted trilayer graphene~\cite{Uri2023}, while eightfold symmetric structures have also been created in optical lattices~\cite{Viebahn19,Sbroscia20}. Quasicrystals host a broad range of exciting physical phenomena from exotic criticality~\cite{Sommers23,Gokmen23,PhysRevLett.131.173402} to charge order~\cite{FlickerVanWezel15,FlickerVanWezel15B} to topology~\cite{TopoQC,MadsenEA13,FlickerVanWezel15B}, with the lack of periodicity often leading to novel behaviors.

Efficient Hamiltonian cycle constructions exist for certain special classes of graph~\cite{gardner1957mathematical,bondy1976method,robinson1992almost}. One example is four-connected planar graphs, defined as requiring the deletion of at least four vertices to disconnect them~\cite{tutte1956theorem}. This result was recently employed in the elegant design of ``quasicrystal kirigami"~\cite{Liu22} using a construction based on Hamiltonian cycles defined on the planar dual to AB. However, to our knowledge AB tilings themselves are not covered by any such special case (they are three-connected). The unexpectedness of (arbitrarily large finite subgraphs of) AB tilings' admittance of Hamiltonian cycles can be seen by comparison to rhombic Penrose tilings, which are in many ways similar to AB. These, too, are a set of infinite 2D tilings built from two tile types (two shapes of rhombus); they are again aperiodic but long-range ordered; and the resulting graphs are bipartite, meaning the vertices divide into two sets such that vertices in one set connect only to vertices in the other -- another property shared with AB. Yet Penrose tilings cannot admit Hamiltonian cycles, because they do not admit perfect dimer matchings (sets of edges such that each vertex meets precisely one edge)~\cite{flicker2020classical}. The latter is a necessary condition for the former, since deleting every second edge along a Hamiltonian cycle results in a perfect matching. It is difficult to think of any special case which could cover AB tilings which would not also cover Penrose tilings. 

The difficulty of constructing Hamiltonian cycles in arbitrary graphs can be made precise using the notion of computational complexity. Given a graph, the question of whether it admits a Hamiltonian cycle lies in the complexity class ``nondeterministic polynomial time complete" (NPC)~~\cite{karp1972reducibility,garey1979computers}. These problems are prohibitively hard to solve -- the fastest known exact algorithms scale exponentially -- but a given solution can be checked quickly, in polynomial time (P). Completeness refers to the fact that if a polynomial-time algorithm were found which solved any NPC problem, all problems in the broader class NP would similarly simplify. Loosely, all problems in NP contain a bottleneck in NPC.

We do not purport to solve an NPC problem. Rather, we show that AB is a previously unknown special case of the Hamiltonian cycle problem which lies in P rather than NPC. As such, it does not follow that all NP problems can be solved in polynomial time on AB. However, we might reasonably hope that our Hamiltonian cycles permit new solutions to certain problems on AB, and indeed we find that this is the case. Using both the Hamiltonian cycles and the inherited discrete scale symmetry of AB tilings, we find exact solutions on AB to a range of other nontrivial problems which are NPC in general graphs~\cite{garey1979computers}. We present three cases which have important applications to physics.

For example, by solving the equal-weight traveling salesperson problem (Sec.~\ref{subsec:impossible1}) on AB we provide a maximally efficient route for a scanning tunneling microscopy (STM) tip to follow in order to scan every atom on the surface of an AB quasicrystal in a given area. Given that a state of the art STM measurement might take on the order of a month, such efficient routes are highly desirable. The problem of finding them is not present in periodic crystals, where efficient routes can be found by following crystal symmetries.

By solving the longest path problem (Sec.~\ref{subsec:impossible2}) on AB we identify the longest possible path visiting every site within a given region. This shows that long, flexible molecules, such as polymers, can adsorb onto the surface of AB quasicrystals with maximal efficiency. Moreover, by breaking the longest path into smaller units, our solution demonstrates that collections of flexible molecules of arbitrary lengths can also adsorb at maximal packing density. This has potential applications to catalysis, in which the activation energy of a reaction between molecules is lowered by having them adsorb onto a surface. Quasicrystal surfaces are emerging as an interesting material for adsorption because they provide different local environments with a range of bond angles~\cite{mcgrath2010surface}, in contrast to the periodic surfaces of crystals. They have potential applications in hydrogen adsorption and storage~\cite{dubois2000new,kweon2022quantitative}, low friction machine parts, and nonstick coatings~\cite{dubois2000new,diehl2008gas}.

The three coloring problem asks whether it is possible to color the tiles of a planar tiling with three colors such that no two tiles sharing an edge share a color~\cite{baxter2016exactly}. We solve this problem on AB in Sec.~\ref{subsec:impossible6}. The three colors can represent the degrees of freedom in the three-state Potts model of nearest neighbor magnetic interactions. The three-colored AB tiling thereby gives an unfrustrated ground state for the antiferromagnetic $q$-state Potts model, with spins defined on the faces of AB, for any $q\ge 3$. 

The Potts model was originally introduced to describe interactions between magnetic ions~\cite{wu1982potts,baxter2016exactly}, gaining relevance to quasicrystals with the  discovery of ordered magnetic states in these materials~\cite{Hattori95,Islam98,Sato00,Goldman13,Tamura21}. More recently it, too, has be used as a simple model of protein folding, with Potts variables representing the $q$ species of amino acid comprising the protein chain~\cite{garel1988mean,Li1996,de2017statistical}. The interaction between amino acids can be either repulsive or attractive depending on the solvent which surrounds the protein. Compared to the typical approximation that the protein folds on a square lattice \cite{chan1989compact,bodroza2013enumeration,de2017statistical}, AB arguably defines a better approximation to the real problem, in which the amino acids have continuous spatial degrees of freedom. This is because AB has a larger range of bond angles, better approximating a continuous rotational symmetry than any 2D regular lattice.

In Appendix~\ref{app:appendixB} we solve a further three problems, on AB, which are known to be NPC in general graphs: the minimum dominating set problem (Sec.~\ref{subsec:impossible3}); the domatic number problem (Sec.~\ref{subsec:impossible4}); and the induced path problem (Sec.~\ref{subsec:impossible5}). Doubtless many more examples can be found. The fact that many NPC problems admit polynomial-time solutions on AB suggests that the discrete scale symmetry of quasiperiodic lattices may be as powerful a simplifying factor as the discrete translational symmetry widely used to obtain exact results on periodic lattices.

The remainder of this paper is organized as follows. We introduce the necessary background on the AB tilings and graph theory in Sec.~\ref{sec:background}. In Sec.~\ref{sec:proof} we prove the existence of Hamiltonian cycles on arbitrarily large finite subgraphs of AB tilings. We prove along the way the possibility of fully packed loops, and give details of our algorithm. In Sec.\ref{sec:solutions_maths} we utilise the approach to present exact solutions, \textrm{on the same graphs}, to three problems which are NPC for general graphs: the equal-weight travelling salesperson problem (Sec.~\ref{subsec:impossible1}); the longest path problem (Sec.~\ref{subsec:impossible2}); and the three-coloring problem (Sec.~\ref{subsec:impossible6}). We comment on the applications for each. In Appendix \ref{app:appendixB} we provide exact solutions to three other problems on these graphs: the minimum dominating set problem (Sec.~\ref{subsec:impossible3}); the domatic number problem (Sec.~\ref{subsec:impossible4}); and the induced path problem (Sec.~\ref{subsec:impossible5}). In Sec.~\ref{sec:conclusions} we place these results in a broader context.

\begin{figure*}[t]
\includegraphics[width=0.9\textwidth]{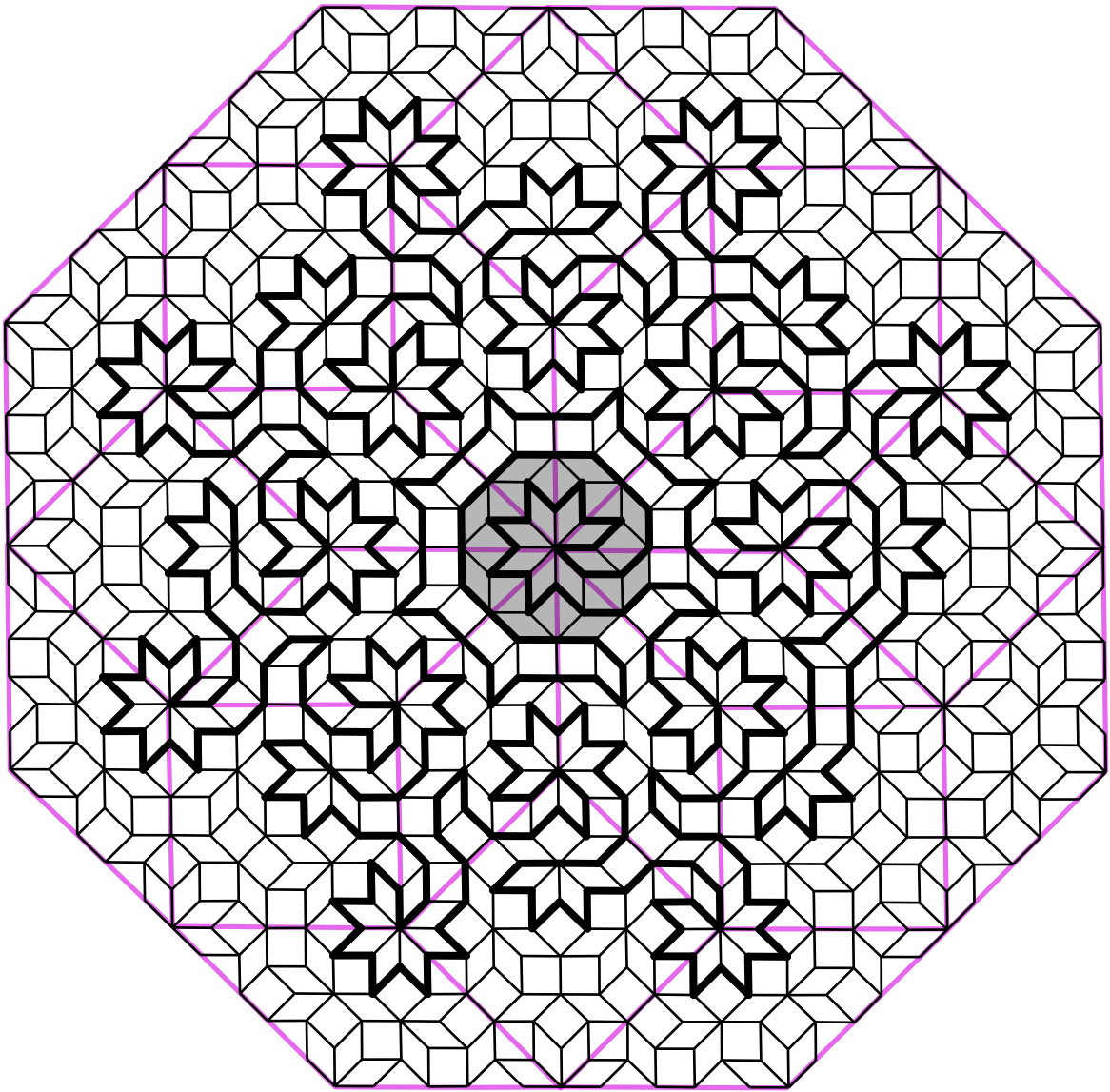}
\caption{A patch of AB tiling showing the 8-fold symmetric region $W_{1}$, and $W_{0}$ in grey. $W_1$ is the twice-inflation of $W_0$, formed from the double-inflated tiles shown in Fig.~\ref{fig:inflation1}; purple lines show the boundary of these tiles (and do not belong to the tiling itself). The thick black edges form a Hamiltonian cycle on the set of vertices they visit; this set is termed $U_1$.}
\label{fig:AB_H_cycle}
\end{figure*}

%
\section{Background}
\label{sec:background}
%

%
\subsection{Ammann Beenker Tilings}
\label{subsec:AB}

The quasicrystal stuctures represented by the Ammann-Beenker tilings lack the translation symmetry of periodic graphs and accordingly, mathematical results in the thermodynamic limit are more challenging to obtain. However, the translation symmetry of ordinary crystals is replaced by a \emph{discrete scale symmetry} which underlies the unique features of quasicrystal systems. 

\begin{figure}[h]
   \includegraphics[width = \linewidth]{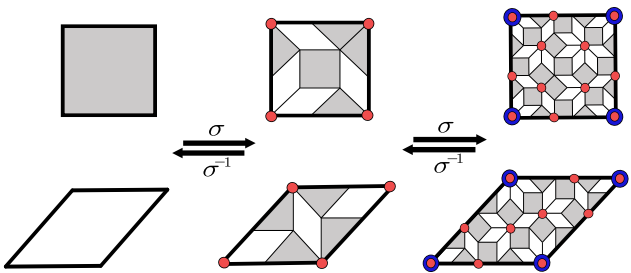}
   \caption{The `inflation' rule $\sigma$, in which a tile is decomposed into smaller tiles. Typically this would be followed by a rescaling (inflation) of all lengths by a factor of the silver ratio $\delta_S = 1+\sqrt{2}$, although as we are only concerned with graph connectivity, inflation and decomposition are equivalent. Vertices of once- and twice-inflated tiles are shown in red and blue respectively. We denote the once-inflated tiles to be level $L_{1/2}$ and the twice-inflated tiles to be level $L_1$.}
   \label{fig:inflation1}
\end{figure}

We will first describe the construction of the AB tilings in the thermodynamic limit~\cite{baake2013aperiodic}, and then discuss the scale symmetry of the tilings, which is central to our proof of Hamiltonian cycles. 

Each tiling is built from two basic building blocks, called `prototiles': a square tile and a rhombic tile with acute angle $\pi/4$. Both tiles are taken to have unit edge length. Starting from any `legal' patch of a few tiles\footnote{Legal patches can be defined according to edge matching rules~\cite{baake2013aperiodic}, and simply ensure that the starting patch corresponds to a group of tiles that can be found together within the infinite tiling. In practice, usually a single square or rhombus tile is used as the initial patch.}, the tiling is then built by repeatedly applying an `inflation rule' to the tiles: every tile is first `decomposed' into smaller copies of the two tiles, as defined in Fig.~\ref{fig:inflation1}, and then the new edges are rescaled (inflated) by a factor of the silver ratio, $\delta_S = 1+\sqrt{2}$, so that the tiles have unit edge length again. The edges and vertices of the tiling define a graph. The graph is bipartite, meaning the vertices divide into two subsets where edges of the graph only join members of one subset to the other. Since the arguments presented in this paper only rely on the connectivity of the graph, the rescaling is not important for our purposes --- hence inflation and decomposition are interchangeable in what follows. Under inflation, the number of tiles grows exponentially and the infinite tiling is recovered in the limit 
\begin{equation}
	\mathcal{T} = \lim_{n\to \infty} \sigma^n(\mathcal{T}_0),
\end{equation}
where $\mathcal{T}$ is the infinite tiling, $\sigma$ the inflation rule, and $\mathcal{T}_0$ the initial patch. We use ``tiling" to mean the infinite tiling, and ``patch'' when referring to any finite set of connected tiles. The inverse process, termed `deflation', equivalently `composition' (followed by a rescaling), is uniquely defined and follows from the inflation rules.

The set of Ammann-Beenker tilings that can be created by inflation has infinite cardinality. There are two important consequences that follow from the inflation structure~\cite{baake2013aperiodic} which will be important later.

First, taking graph edges to be of unit length:
\begin{definition}[Linear Repetitivity~\cite{baake2013aperiodic,Haynes18}]\label{def:LR}
There is a finite $C>0$ such that, for any set of vertices $V$ appearing within a disk of radius $r$, every disk of radius $Cr$ contains a copy of $V$.
\end{definition}
AB is linearly repetitive~\cite{baake2013aperiodic}.

\begin{definition}[Local Isomorphism~\cite{GrunbaumShephard}]\label{def:LI}
Two tilings are locally isomorphic if any finite patch appearing in one appears in the other. 
\end{definition}
All AB tilings are locally isomorphic to one another~\cite{GrunbaumShephard,baake2013aperiodic}. In this paper `AB' refers to the set of all AB tilings unless otherwise stated. As a result of linear repetitivity and local isomorphism, it should be understood that reference to a vertex set $V$ in a particular AB tiling is meant modulo translations.

Within the locally isomorphic set we will not differentiate between arbitrary rescalings of the tilings. As the simplest example of this isomorphism, the set of edges and tiles surrounding each vertex in the tiling belongs to one of seven unique configurations, which we call `$m$-vertices' (where $m$ labels the vertex connectivity). In total the vertex configurations for AB are given by the $3$-, $4$-, $5_A$-, $5_B$-, $6$-, $7$- and $8$-vertices, with each configuration appearing with a frequency given by a function of the silver ratio $\delta_S$. The $3$-vertex, for example, occurs most frequently and makes up a fraction $\delta_S-2$ ($\sim 40\%$) of the AB vertices. The two configurations with five edges are distinguished by their behaviour under inflation. 

A special role is played by the 8-vertices of the tiling, as every vertex configuration inflates to an 8-vertex under at most two inflations~\cite{lloyd2021statistical}. This means that edges can be drawn between 8-vertices of an AB tiling so as to generate another AB tiling. Fig. ~\ref{fig:AB_H_cycle} highlights this symmetry. Similarly, some of the 8-vertices of the original tiling sit also at 8-vertices of the composed tiling, and therefore logically form the \emph{four-times} composed tiling. This hierarchy continues and results in the discrete scale symmetry exhibited by the AB tiling and several other aperiodic tilings, in turn responsible for many of the remarkable physical properties of quasicrystals. 

We call an $8$-vertex an $8_0$-vertex if under twice-deflation it becomes any $m$-vertex with $m\neq 8$ (most 8-vertices of the AB tiling are of this type); similarly we call it an $8_1$-vertex if it becomes an $8_0$-vertex under twice-deflation. Generalising, an $8_n$-vertex becomes an $8_0$-vertex after $2n$ deflations. The \emph{local empire} of a vertex configuration is the simply connected set of tiles that always appears around the vertex configuration wherever it appears in the tiling~\cite{fang2018non}. The local empire of the $8_0$-vertex, which we call $W_0$, is shown in Fig.~\ref{fig:AB_H_cycle} (also shown is the $8_1$ empire, $W_1$): it has a discrete 8-fold rotational symmetry ($D_8$ in Sch\"{o}nflies notation). The inflation rule maps $8_n$-vertices to $8_{n+1}$-vertices (applying the rule twice) while simultaneously growing the size of the vertex's local empire. The $D_8$ symmetry is preserved under inflation. Therefore $8_n$-vertices are accompanied by $D_8$-symmetric local empires $W_n$ having a radius of symmetry growing approximately as $R_n \sim \delta_S^{2n}$. Furthermore the $W_n$ empire contains all $W_{0\leq m\leq n}$ empires. 

It follows from linear repetitivity that any finite set of tiles (or vertices) in AB is contained in $W_n$ for sufficiently large $n$. On account of this, we will often focus on the $W_n$ to prove more general results. It will also prove useful to define a modified version of the AB tiling with all 8-vertices removed, dubbed the AB* tiling~\cite{lloyd2021statistical}.

In Fig.~\ref{fig:inflation1} we denote the individual square and rhombus (left) to be level $L_0$ and the \emph{twice} inflated square and rhombus (right) to be level $L_1$; this is because any vertex inflates to an 8-vertex under two inflations but not one, and the convention follows Ref.~\onlinecite{lloyd2021statistical}. We therefore designate the \emph{once}-inflated tiles (middle) to be at level $L_{1/2}$. 
\subsection{Graph terminology and conventions}
\label{subsec:graph}
A graph $G(V,E)$ is a set of vertices $V$ connected by a set of edges $E$. We consider undirected graphs, in which no distinction is drawn between the two directions of traversal of an edge. We also consider only bipartite graphs, in which the vertices divide into two sets such that edges only connect one set to the other. The cardinality of a graph is the cardinality of the set of vertices it contains (i.e. the number of vertices in the graph). We denote this $|G|$~\cite{tutte2001graph, bondy1976graph, chartrand1977introductory}. 

We define a `path' to be a set of edges joining a sequence of distinct vertices, and the length of the path to be the number of edges it contains. Given a set of edges $E'\subset G$ it may be possible to find an `alternating path', which is a path in $G$ along which every second edge belongs to $E'$ but all others do not~\cite{tutte2001graph}. We borrow this terminology from the theory of dimer matchings, in which $E'$ is a set of edges such that no two edges share a common vertex, although we consider different structures for $E'$ here. An `augmenting path' is an alternating path in which the first and last edges are not in $E'$. In general, `augmenting' an alternating path means to switch which edges along the path are in the set $E'$ and which are not.

In this paper, we consider graphs made from the vertices and edges of Ammann-Beenker tilings. We will use the shorthand $G=$~AB to denote the case that $G$ is the graph formed from any infinite Ammann-Beenker tiling. We will sometimes refer to graphs formed from finite patches of AB, which should be clear from context.

\section{Constructive Proof of Hamiltonian Cycles}
\label{sec:proof}

In this Section we prove the following.
\begin{theorem}\label{thm:H_cycles}
Given an AB tiling and a finite set of vertices $V\subset$~AB there exists a set $U_n$, where $V\subseteq U_n\subset$~AB, such that $U_n$ contains a Hamiltonian cycle $\mathcal{H}$.
\end{theorem}

\begin{proof}\label{proof:H_cycles}
In Section~\ref{subsec:FPL_AB*} we identify a set of edges on the twice-inflated AB tiles such that every $8_{n<0}$ vertex (that is, every vertex which is not an $8$-vertex) meets two such edges. This constitutes a set of fully packed loops (FPLs), visiting every vertex precisely once, on AB*. We then focus on finite tile sets generated by $n$ twice-inflations of the local empire of the $8_0$-vertex, denoted $W_n$. These regions have $D_8$ symmetry. Any finite set of vertices $V\subset$~AB lies within an infinite hierarchy of $W_n$ for sufficiently large $n$, as shown in Section \ref{sec:background}. In Section~\ref{subsec:FPL_AB} we identify a method of reconnecting these fully packed loops so as to include into the loops all $8_0$-vertices within $W_n$, and then all $8_{0<m<n}$ vertices where the $8_{m-1}$ vertices have already been included. The result is fully packed loops on all vertices within $W_{n}$ except the central $8_n$ vertex. Finally, we show that a subset of these loops can be joined into a single loop which additionally visits the central vertex. We denote the set of vertices visited by this loop $U_n$. 
As the patches $U_n$ grow exponentially with $n$ so as to cover arbitrarily large areas of AB, it follows from linear repetitivity (Definition \ref{def:LR}) that $V\in U_n$ for sufficiently large $n$. 

Since the cycle constructed on $U_n$ visits each of its vertices precisely once, this proves Theorem~\ref{thm:H_cycles}.
%
\end{proof}

The proof is constructive, returning $\mathcal{H}$ given $V$, and is linear in the number of vertices in $U_n$.
We place a precise bound on the necessary size of $U_n$ for any given $V$ in Section~\ref{subsec:bound_Un}.
Furthermore, the Hamiltonian cycle on $U_n$ visits a simply-connected set of vertices whose cardinality increases exponentially with $n$, and which therefore admits a straightforward approach to the thermodynamic limit.
We consider this limit, $U_{n\rightarrow\infty}$, in Section \ref{subsec:Thermodynamic_limit}.

\subsection{Constructing Fully Packed Loops on AB*}
\label{subsec:FPL_AB*}

\begin{figure}[h]
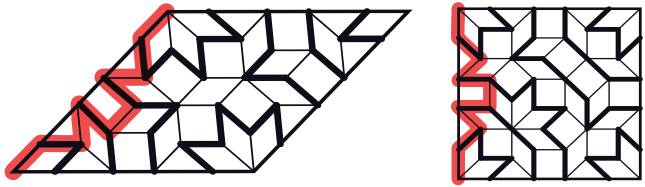

\centering
\includegraphics[width=.31\textwidth]{Rh1.pdf}
\hspace{0.2cm}
\includegraphics[width=.15\textwidth]{Sq1.pdf}
\caption{Twice-inflations of each of the two AB prototiles. We denote the smaller tiles composition-level zero, $L_0$, and the larger tiles $L_1$. The thick black edges visit all but the corner vertices. We call these edges $e_0$ edges. Since the 8-vertices at $L_0$ corresponds to the tile vertices at $L_1$, the union of $e_0$ constitutes fully packed loops on AB* (AB without the 8-vertices). Augmenting the red path (switching covered and non-covered edges) places loop ends on the two corner vertices while still visiting the original vertices. The red path can be thought of as the twice-inflation of an $L_1$ edge, and we term it $e_1$. It has the same effect along any tile edge. Augmenting cycles built from $e_1$ then places all visited 8-vertices onto the same loop.
}
\label{fig:proof}
\end{figure}

In Fig.~\ref{fig:proof} we show the twice-inflations of the two prototiles. We denote the smaller tiles as constituting composition level zero, $L_0$, and the larger tiles $L_1$. We work exclusively with twice-inflations, rather than single inflations, as every vertex becomes an 8-vertex under twice-inflation, but some do not do so under a single inflation. 

In Fig.~\ref{fig:proof} we highlight (thick black edges) a subset of $L_0$ edges on each $L_1$ tile. When the $L_1$ tiles join into legitimate AB patches, every vertex in $L_0$ is visited by precisely two thick edges, except those vertices sitting at the corners of the $L_1$ tiles (and vertices on the patch boundary, which are not important when we take the thermodynamic limit since we deal only with finite subgraphs interior to the region). We call each of these edges an $e_0$ edge. The corner vertices of $L_1$ tiles are exactly the 8-vertices of $L_0$ (this follows from the fact that every vertex becomes an 8-vertex under at most two inflations). Therefore this choice causes every vertex of AB* at $L_0$ to be met by two $e_0$ edges. Since every vertex meets precisely two edges, the sets of edges must form closed loops~\footnote{Strictly, loops can be open if they connect two points on the boundary. However, these open paths lie entirely in the region $W_n\setminus U_n$ so are not relevant to our construction on $U_n$.}. Therefore these loops constitute a set of fully packed loops (FPLs) on AB* as required.

While it is possible to find other sets of edges with these properties, we designate this choice the canonical one. With it, all closed loops respect the $D_8$ local symmetry of AB and AB*.

\subsection{Constructing Fully Packed Loops and Hamiltonian Cycles on AB}
\label{subsec:FPL_AB}

We next seek to construct FPLs on AB rather than AB* by adding the missing 8-vertices onto the loops. In Fig.~\ref{fig:proof} we highlight in red an alternating path (with respect to $e_0$ edges) which connects nearest-neighbor vertices in $L_1$. Augmenting the red path (swapping which edges are covered by $e_0$ edges) has a number of effects. First, it 
places $e_0$ edges touching the two $L_0$ 8-vertices on which the alternating path terminates. Second, it increases the total number of $e_0$ edges by one. Third, any vertex which was previously visited by a pair of $e_0$ edges is still visited by a pair of $e_0$ edges. 

While other such paths are possible, this canonical choice has advantages when combined with the canonical choice of $e_0$ edges. First, it can be seen in Fig.~\ref{fig:proof} that exactly the same shape of path can be used on any edge of either $L_1$ tile to cause the set of three effects just listed. Second, it does so entirely within the tile itself, and so these augmentations can be carried out without reference to neighboring $L_1$ tiles. The augmentation can be undone by augmenting along the same path a second time, which turns out to be key.

The red alternating paths trace a route along $L_0$ graph edges which follows the $L_1$ edges as closely as possible while also alternating with reference to the $e_0$ edge placement. In a natural sense, then, the red path is the twice-inflation of an edge of an $L_1$ tile. We therefore refer to the red path as an $e_1$ edge. In this way, $e_0$ edges connect $8_{n<0}$-vertices (i.e.~any vertices \emph{except} 8-vertices at $L_0$), while $e_1$ edges connect $8_0$-vertices (8-vertices which can survive precisely zero deflations while remaining 8-vertices). Further inflations can be carried out by stitching $e_1$ edges together in exactly the same way that $e_1$ was formed from $e_0$. For example, $e_2$ edges, connecting $8_1$-vertices, can be built from $e_1$ edges, and can equivalently be thought of as built from more $e_0$ edges. In general, $8_n$-vertices can be connected by $e_{n+1}$ edges. The first three levels are shown in Fig.~\ref{fig:dimer_inflations}.

In the limit $n\rightarrow\infty$, $e_n$ is a fractal: under $e_0\rightarrow e_1$ the initial side length of any prototile is divided into 9 segments each of length $\delta_S^{-2}$. The same scaling occurs for all subsequent inflations, and so the box counting dimension of $e_\infty$ is given as 
\begin{align}
\textrm{dim}(e_\infty)=\lim_{n\rightarrow\infty}\frac{\log\left(9^n\right)}{\log\left(\delta_S^{2n}\right)}=\frac{1}{\log_3\left(1+\sqrt2\right)}\approx 1.246.
\end{align}
This fits with the intuition that the curve is space-filling.

\begin{figure}[h]
\includegraphics[width=.4\textwidth]{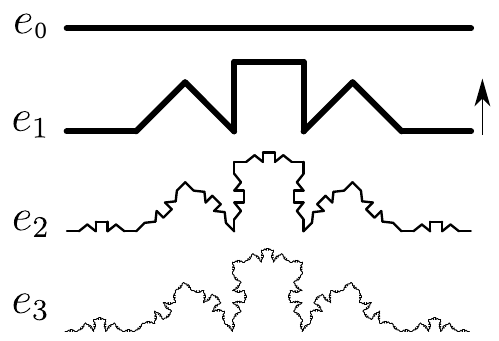}
\caption{The red alternating path in Fig.~\ref{fig:proof} can be thought of as a twice-inflation of an $e_0$ edge from which loops are constructed at level zero ($L_0$). It connects nearest neighbors in the $L_1$ tiling (8-vertices in the $L_0$ tiling) while following an alternating path in the canonical choice of $e_0$ placements. We therefore denote this path an $e_1$ edge. Alternating paths connecting higher-order $8_n$-vertices can be constructed by concatenating $e_{n+1}$ edges in the same way that $e_1$ was formed from a concatenation of $e_0$; levels $e_2$ and $e_3$ are shown here. We define the orientation of $e_n$ according to the arrow.
}
\label{fig:dimer_inflations}
\end{figure}

\begin{figure*}[t]
\centering
\includegraphics[width=15cm]{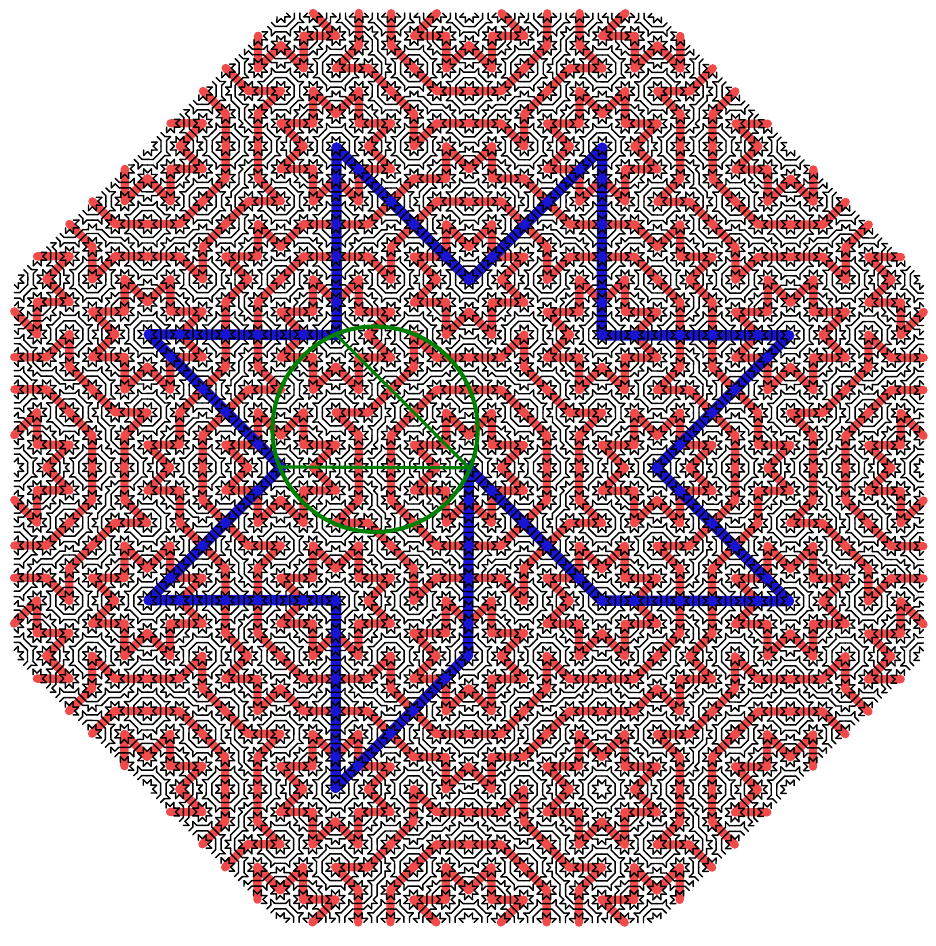}
\caption{The canonical placements of $e_0$ edges (black) form $e_0$-loops visiting every $8_{n<0}$-vertex (i.e.~every vertex which is not an $8$-vertex). The red lines form $e_1$-loops, along which $e_1$ edges from Fig.~\ref{fig:dimer_inflations} can be placed adding all $8_0$-vertices onto loops. Similarly the blue star forms an $e_2$-loop, along which $e_2$ edges from Fig.~\ref{fig:dimer_inflations} can be placed adding $8_1$-vertices onto loops (see Fig.~\ref{fig:AP1}). The central 8-vertex is now added onto loops by folding one corner of the blue star inwards to break the 8-fold symmetry. Here we have shown the process explained in Sec.~\ref{subsec:FPL_AB} to order $n=2$, but it can be iterated to any order $n$. 
In green is the largest disk fitting within the blue loop; the disk connects three 8-vertices with green chords.
\label{fig:thermodynamic_limit}
}
\end{figure*}

To complete the proof of the existence of an FPL on AB it remains to show that all 8-vertices can be placed onto closed $e_n$-loops for sufficiently large $n$. The canonical edge covering places all $8_{n<0}$ vertices of AB onto $e_0$-loops. To add all $8_0$-vertices onto loops, we place loops of $e_1$ edges according to the canonical edge placement. This placement was defined at level $L_0$; but note that the twice-(de)composition of any AB tiling is another AB tiling. Therefore the placement is well defined at all levels. Placing an $e_1$ edge means augmenting a path at $L_0$. Augmenting the closed loops of paths just defined places all $8_0$-vertices visited by that loop onto loops of $e_0$ edges as required. We then proceed by induction, adding $8_n$-vertices by connecting them with loops of $e_{n+1}$ edges. After $n$ steps, every vertex of order $8_{\le n}$ is contained in a set of fully packed loops defined on those vertices. By design, placing the $e_n$ edges does not cause problems with the existing matching of $e_{n-1}$ edges. We refer the reader to Fig.~\ref{fig:thermodynamic_limit} for an illustration of the $e_1$ loop structure (red), which clearly mirrors the smaller $e_0$ loop structure (black) constructed in Sec.~\ref{subsec:FPL_AB*}.

While $e_1$ edges must follow the edges of $L_1$ tiles, two tiles meet along any edge. There is therefore a choice of two orientations of each $e_1$ owing to the choice of which of the two tiles $e_1$ lies within. Each $e_1$ orientation can be chosen freely with the statements in the previous paragraph remaining true. However, we find that there is again a natural choice. Define the orientation of $e_1$ to be in the direction indicated in Fig.~\ref{fig:dimer_inflations}. At $L_1$, we choose the sequence of $e_1$ edges along a loop to point alternately into then out of the loop on which it sits (this is implicit in the construction of $e_2$ and $e_3$ in Fig.~\ref{fig:dimer_inflations}). Since AB is bipartite, all loops are of even length, and there is never an inconsistency. When $e_2$ edges are placed, these will cut through $e_1$-loops. It is always possible to choose the orientations so that wherever $e_1$-loops and $e_2$-loops intersect, they do so along the length of one $e_1$ (recall that an $e_2$ is built from multiple $e_1$). Since these two $e_1$ overlap perfectly, augmenting them both has no overall effect at $L_0$. We can therefore delete both $e_1$. The effect is that the $L_2$ loops no longer intersect the $L_1$ loops: instead, they rewire separate $L_1$ loops so as to make them join together. Augmenting the $L_1$ star rewires all `cut' loops into one. By induction the rewiring works at all levels. The process is shown in detail in Appendix A, Fig.~\ref{fig:AP1} and Fig.~\ref{fig:AP2}.

\begin{figure*}[t]
    \centering
    \includegraphics[width=18cm]{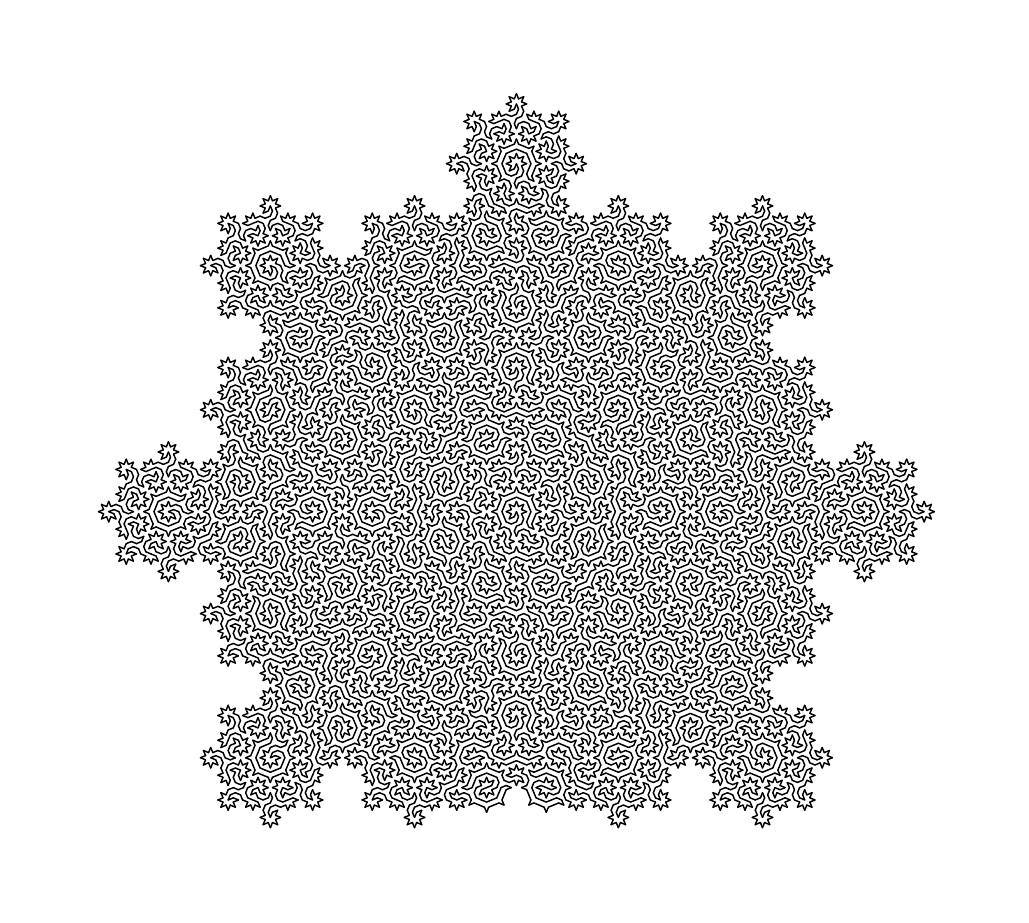}
    \caption{The Hamiltonian cycle visiting all the vertices of a $U_2$ region (the AB tiles themselves are omitted for clarity). The image is obtained from Fig.~\ref{fig:thermodynamic_limit} by placing $e_1$ and $e_2$ edges in alternate orientations along the red and blue loops respectively, as shown in Fig.~\ref{fig:AP1}, and augmenting. Note the present figure has been rotated through $1/16^{\textrm{th}}$ of a turn relative to other figures to utilise the page efficiently.}
    \label{fig:giant_H_cycle}
\end{figure*}

%
%
While it is already possible to reconnect many of these loops together, thereby building Hamiltonian cycles on subsets of the $W_n$ vertices, the sets of vertices visited by such cycles always encircle sets of vertices they do not visit. To remedy this, it is necessary to break the $D_8$ symmetry. To see that this is necessary, consider the simplest $D_8$ Hamiltonian cycle, which is a star visiting the eight vertices adjacent to any 8-vertex, and the eight vertices immediately beyond them. To include the central 8-vertex on a loop, one of the points of the star must turn inwards, breaking the symmetry. 

In Fig.~\ref{fig:thermodynamic_limit} we show a simple way to include the central region. The outermost $e_n$-loop encircling the $8_n$-vertex is the $2n$-fold inflation of the smallest star loop. By folding in a single corner of this inflated star, the loop now visits the central $8_n$-vertex. In so doing, it connects \emph{every} loop contained within it (and all those it passes through) into a single loop. The result is therefore a Hamiltonian cycle on the simply-connected set of vertices visited by this deformed star. We denote this set $U_n$ to distinguish it from the $D_8$ symmetric set $W_n$. The cycle on $U_2$ is shown in Fig.~\ref{fig:giant_H_cycle}.

There is a topological reason this construction must work. Every closed $e_0$ loop bounds a $D_8$-symmetric region centred on an 8-vertex. Topologically we can imagine deforming the $e_0$ loops however we like, provided they cross neither 8-vertices nor other loops. Fig.~\ref{fig:proof} shows that any path connecting two $L_1$ corners (8-vertices) along an $L_1$ edge must cross exactly four $e_0$ loops. One of these loops closes around one of the 8-vertices, while the three other loops close around the other 8-vertex (as well as other 8-vertices). This follows from Fig.~\ref{fig:proof}, but a detailed proof appears in Appendix D of Ref.~\onlinecite{lloyd2021statistical}. The situation is shown schematically in Fig.~\ref{fig:topology}A. Augmenting the red $e_1$ path along the $L_1$ edge  rewires the $L_0$ loops according to Fig.~\ref{fig:topology}B. The same rewiring occurs for the blue edge in Fig.~\ref{fig:topology}C. While the augmentation illustrated in Fig.~\ref{fig:topology} leaves two string ends free, augmenting the closed $L_1$ star loop leads to a single closed loop at $L_0$. All subsequent levels of inflation $L_n$ follow precisely the same pattern of connectivity. Having $e_1$ orientations (arrows in Fig.~\ref{fig:topology}, defined in Fig.\ref{fig:dimer_inflations}) point alternately in and out of the loop introduces a chirality that guarantees two facts. First, whenever an $e_{n+1}$-loop cuts through $e_n$-loops, the $e_n$-loops merge. Second, the folded-corner outermost loop connects \emph{all} loops within it. While this might at first seem miraculous, really it is no more complicated than the (folded) 16-vertex $L_0$ star hitting all 16 vertices in $U_0$: the $U_n$ Hamiltonian cycle simply has $n$ inflations (of both tiles and loops) applied according to the rules of Fig.~\ref{fig:proof}.

\subsection{Bounding the size of $U_n$}\label{subsec:bound_Un}

To see that any vertex set $V$ must be contained within some $U_n$, recall linear repetitivity: any $V$ lying within a disk of radius $r$ must appear within \emph{all} disks of radius $Cr$. Hence, any disk of radius $Cr$ must contain all vertex configurations that can be contained within a disk of radius $r$. 

Fig.~\ref{fig:thermodynamic_limit} shows a green disk contained within $U_2$. This is the largest disk fitting within the blue loop. Larger disks fit within $U_2$, but the boundary of successive $U_n$ becomes fractal, whereas the equivalent to the blue loop maintains its shape at all levels of inflation, making the subsequent statements simpler. Denoting the green chord lengths $l_2$ (noting that these are edges of an $L_2$ rhombus) the radius of this disk is $2^{-1/4}\delta_S^{-1/2}\,l_2$. The largest equivalent disk in region $U_n$ similarly has radius $2^{-1/4}\delta_S^{-1/2}\,l_n$. Denoting $e_0$ to be of unit length, the inflation rules of Fig.~\ref{fig:inflation1} imply that $l_n=\delta_S^{2n}$, and so $U_n$ contains a disk of radius $2^{-1/4}\delta_S^{2n-1/2}$. Hence, linear repetitivity implies that every $U_n$ contains every vertex set contained within a disk of radius $2^{-1/4}\delta_S^{2n-1/2}/C$. This radius grows exponentially and without bound in $n$, and so any $V$ is contained within $U_n$ for sufficiently large $n$ (indeed, it is contained within all $U_{n\ge M}$ for sufficiently large $M$).

\subsection{Thermodynamic limit}\label{subsec:Thermodynamic_limit}

On any $U_n$ we can construct a Hamiltonian cycle by the method outlined in Sec.~\ref{subsec:bound_Un}. Since for any finite vertex set $V$ we can find a $U_n$ such that $V \subseteq U_n$, the $U_n$ admit a natural extension to the thermodynamic limit of macroscopically large cycles, $U_{n\rightarrow\infty}$ (meaning arbitrarily large finite patches). The algorithm is of linear complexity in the number of vertices within $U_n$.

\begin{figure}[h]
\includegraphics[width=\columnwidth]{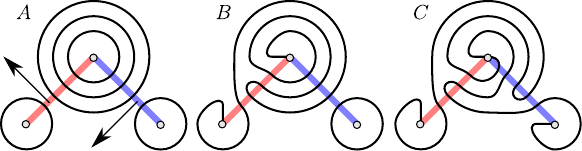}
\caption{The topology underlying the Hamiltonian cycle construction. $A$: every 8-vertex (grey disks) separated from another by an $L_1$ edge (red/blue) is enclosed by $e_0$ loops (black) as shown (see Fig.~\ref{fig:proof}). Other 8-vertices are omitted for clarity. $B$ ($C)$: augmenting the red (blue) $e_1$ edge, with $e_1$ orientation indicated with arrows as in Fig.~\ref{fig:dimer_inflations}, rewires the loops as shown. Hence, augmenting a \emph{loop} of $e_1$ edges results in a \emph{single} $L_0$ cycle visiting the chosen 8-vertices. The argument holds at all levels of inflation.}
\label{fig:topology}
\end{figure}

The region $U_\infty$, without any surrounding tiles, can itself be thought of as an infinite AB tiling, in the sense that it is an infinite simply connected set of tiles generated by repeated inflations of a legitimate AB patch (followed by deletion of the surrounding tiles in $W_\infty$). 

It should be noted that an AB tiling can have at most one global centre of $D_8$ symmetry (more than one would violate the crystallographic restriction theorem~\cite{AshcroftMermin}), and the set of tilings with a global $D_8$ centre is measure zero in the set of all possible AB tilings~\cite{baake2013aperiodic}. However, our definition of the thermodynamic limit does not rely on our tiling having a global $D_8$ centre -- it only needs to contain at least one 
%
$U_{n\rightarrow\infty}$, and all AB tilings contain infinitely many.
%

%
\section{Solutions to other non-trivial problems on AB}
\label{sec:solutions_maths}
%

The question of whether an arbitrary graph admits a Hamiltonian cycle lies in the complexity class NPC. Problems in NPC are decision problems, meaning they are answered either by yes or no. The optimisation version of these problems -- to provide a solution if one exists -- is instead typically in the class NP-Hard (NPH). Finding a Hamiltonian cycle on an arbitrary graph is NPH. 

In this section we provide exact solutions on AB for three problems made tractable due to discrete scale symmetry and/or by our construction of Hamiltonian cycles. In each case we state the decision problem and the corresponding optimisation problem. Our choice of problems is motivated by the fact that on general graphs these decision problems lie in the complexity class NPC, and the corresponding optimisation problems lie in the complexity class NPH. By solving the problems on AB we again show that this setting provides a special, simpler case.

\subsection{The Equal-Weight Travelling Salesperson Problem}
\label{subsec:impossible1}

\textbf{Problem statement~\cite{garey1979computers, shmoys1985traveling}:} given a number of cities $N$, unit distances between each pair of connected cities, and an integer $k$, does there exist a route shorter than $k$ which visits every city exactly once and returns back to the original city? The corresponding optimisation problem is to find such a route.

\textbf{Solution on AB:} if cities are the vertices of $U_{n\geq 0}$, yes iff $k>|U_{n}|$. 

\textbf{Proof:} from a graph theory perspective, we can consider every city as a vertex and every direct route between a pair of cities as a weighted edge, where the associated weight denotes the distance between those cities. Finding the shortest route which can be taken by the travelling salesperson is equivalent to finding the lowest-weight Hamiltonian cycle, where the weight of the cycle is the sum of the weights of its edges. 

In this paper we have considered the unweighted AB graph, equivalent to setting all edge weights to one. In unweighted graphs the problem reduces to the equal-weight travelling salesperson problem. After this simplification, the unweighted decision problem becomes equivalent to the Hamiltonian cycle problem. Therefore the Hamiltonian cycle constructed in Section~\ref{sec:proof} solves the equal-weight travelling salesperson optimisation problem on $U_n$.

\subsubsection*{Application: scanning microscopy}
\label{subsubsec:STM}

One physical application of the travelling salesperson problem is to find the most efficient route to scan an atomically sized tip across a surface so as to visit every atom. For example, scanning tunneling microscopy (STM) involves applying a voltage between an atomically sharp tip and the surface of a material. Electrons tunnel across the gap between the tip and the sample, giving a current proportional to the local density of states in the material under the tip~\cite{binnig1983scanning, hansma1987scanning}. Magnetic force microscopy (MFM) uses a magnetic tip to detect the change in the magnetic field gradient~\cite{martin1987magnetic, hartmann1999magnetic}. The ultra-high resolution nature of the imaging means that state of the art measurements might take on the order of a month to scan a $100\,$nm $\times$ $100\,$nm square region of a surface~\footnote{State of the art STM measurements can obtain intra-unit-cell data on crystal surfaces; for a $100\times100\,$sq-nm region this would exceed 250,000 pixels. A typical measurement takes around $100\,$ms, so if each pixel is scanned at 100 energies, the result takes around 29 days. We thank J.~C.~S.~Davis for these estimates.}. While generally applied to periodic crystalline surfaces, STM and MFM can in principle be used to image the surfaces of aperiodic quasicrystals, including those with the symmetries of AB tilings~\cite{wang1987two}. Unlike in the crystalline case, the most efficient route for the STM tip to visit each atom is not obvious in these cases. Our solution to the TSP optimisation problem on AB provides a maximally efficient route for aperiodic quasicrystals with the symmetries of AB tilings. While the surface would need to be scanned once in order to establish the $U_n$ regions to study, the purpose of STM and MFM would be to detect changes in the material under changing conditions (say, temperature or magnetic field), and the route provides maximum efficiency upon multiple scans.

\subsection{The Longest Path Problem}
\label{subsec:impossible2}

\textbf{Problem statement~\cite{garey1979computers, karger1997approximating}:} given an unweighted graph $G$ and an integer $k$, does $G$ contain a path of length at least $k$? We emphasise that the path may not revisit vertices. The corresponding optimisation problem is to find a maximum length path.

\textbf{Solution:} yes, if $G=U_{n\geq 0}$ and $k\leq|U_{n}|-1$.

\textbf{Proof:} a Hamiltonian path can be obtained by removing any edge from a Hamiltonian cycle. The results of Section \ref{subsec:impossible1} imply that there exists a path of length $|U_n|-1$ in any region $U_n$. 

\textbf{Comment:} shorter paths can be found by deleting further contiguous edges. Since any AB tiling contains regions $U_n$, for any $n$, it follows that the tiling contains paths of any length. 

\subsubsection*{Application: adsorption}
\label{subsubsec:adsorption}

The `dimer model' in statistical physics seeks sets of edges on a graph such that each vertex connects to precisely one edge. It was originally motivated by understanding the statistics and densities of efficient packings of short linear molecules adsorbed onto the surfaces of crystals~\cite{Roberts1935,fowler1937attempt,Heilmann_lieb}. Noting again that certain physical quasicrystals have the symmetries of AB tilings on the atomic scale, the existence of a longest path visiting every vertex shows that a long flexible molecule such as a polymer could wind so as to perfectly pack an appropriately chosen surface of such a material. The path can be broken into segments of any smaller length, showing that flexible molecules of arbitrary length can pack perfectly on to the surface. Dimers are returned as a special case.

Adsorption has major industrial applications. In catalysis, for example, reacting molecules can find a reaction pathway with a lower activation energy by first adsorbing onto a surface. While efficient packings can be identified on periodic crystals such as the square lattice, these feature a limited range of nearest-neighbor bond angles (with only right angles appearing in the square lattice). Realistic molecules, which have some degree of flexibility, might do better on quasicrystalline surfaces which necessarily contain a range of bond angles (four in AB). Other uses could include (hydro)carbon sequestration and storage, and protein adsorption~\cite{Roberts1935,Heilmann_lieb}.

\subsection{The Three-coloring Problem}
\label{subsec:impossible6}

\textbf{Problem statement~\cite{garey1974some, garey1979computers}:} can all the faces of a planar graph $G$ be colored such that no faces sharing an edge share a color? The corresponding optimisation problem is to find such a three-coloring.

\textbf{Solution:} yes, if $G\subseteq$~AB.

\textbf{Proof:} we know that the AB tiling is a bipartite graph which means that its vertices can be partitioned into two disjoint sets such that none of the edges has vertices belonging to the same set. We associate two opposite bipartite `charges' to the vertices depending upon which of these two mutually exclusive sets they belong to.

We define two three-colored tiles, 1 and 2, as shown in Figs.~\ref{fig:3coltiles}(a) and (b), and repeat them over the entire tiling in the following manner.
\begin{itemize}
    \item Place tile 1 over any 8-vertices having the same bipartite charge. 
    \item Place the mirror image of tile 1 over any 8-vertices having the opposite bipartite charge.
    \item After that, place tile 2 or its mirror image such that the three rhombuses around every ladder, shown in Fig.~\ref{fig:3coltiles}(c), are the same color.
\end{itemize}
The consistent placement of these two tiles on the whole tiling is ensured by the structure of the tiling itself. After filling the whole AB tiling with these two tiles all that remain are ladders, as shown in Fig.~\ref{fig:3colSQ}, which can be colored consistently on the basis of colors of their surrounding rhombuses such that no two adjacent faces should have a same color. 

\begin{figure}[h]
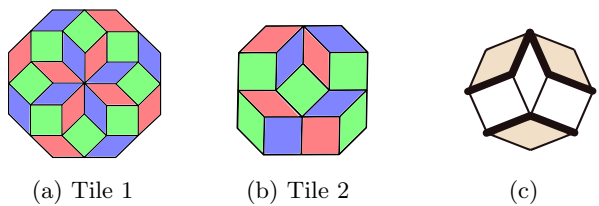

    \centering
    \subfloat[\centering Tile 1]{{\includegraphics[width=2.2cm]{tile_1_1.pdf}}}
    \hspace{0.7cm}
    \subfloat[\centering Tile 2]{{\includegraphics[width=1.9cm]{tile_2_1.pdf}}}
    \hspace{1cm}
    \subfloat[\centering]{{\includegraphics[width=2cm]{tile.pdf}}}
    \caption{(a),(b) Tiles used for 3-coloring of the AB tiling. (c) The structure with three rhombuses present around every ladder.}
    \label{fig:3coltiles}
\end{figure}

The 3-coloring solution for the AB tiling is shown in Fig.~\ref{fig:threecol}. Note that the three-colorings are not unique.

\begin{figure}[h]
    \centering
    \includegraphics[width=8cm]{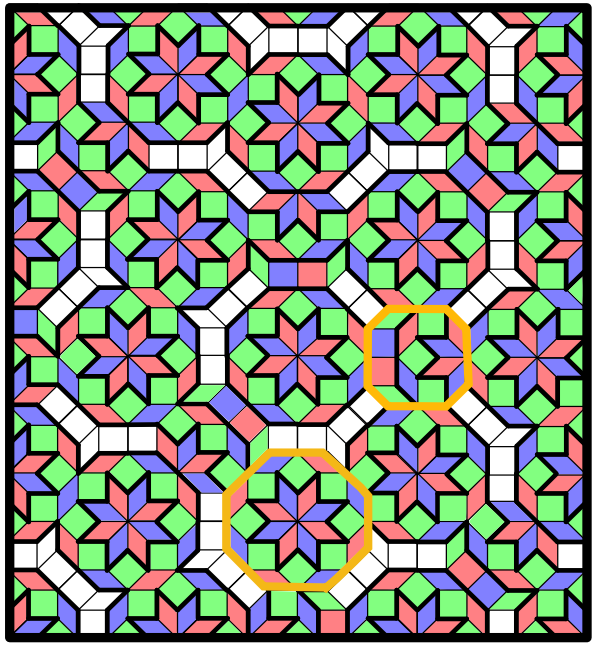}
    \caption{A small patch of AB tiling 3-colored using tile 1 and tile 2 (both outlined in gold) as mentioned in Sec \ref{subsec:impossible6}. The only remaining portions of the tiling are segments of ladders which can be 3-colored consistently on the basis of their surrounding tiles.}
    \label{fig:3colSQ}
\end{figure}

\textbf{Comments:} this proof does not rely directly on the existence of Hamiltonian cycles. Rather, it relies on an intermediate step in their construction (specifically, the existence of fully packed loops on AB*, Section \ref{subsec:FPL_AB*}), as well as the discrete scale invariance of AB. This proof demonstrates that the intermediate steps in our Hamiltonian cycle construction can already unblock problems that lie in the complexity class NP in general graphs, without necessarily needing to appeal to the finished result.

As a historical note, the five-coloring theorem, which states that any political map requires at most five colors in order to avoid any neighboring countries being the same color, was proven in the 19th century~\cite{heawood1890map}. The equivalent four-coloring theorem was an infamous case of a problem which is easy to state but difficult to solve. The eventual proof in 1976 was the first major use of theorem-proving software, and no simple proof has been forthcoming~\cite{AppelHaken}. The three-coloring problem is known to be NPC~\cite{garey1979computers} on general graphs. However, special cases can again be solved. A polynomial-time algorithm for generating a three-coloring of the rhombic Penrose tiling was deduced in 2000, following its conjectured existence by Conway\cite{sibley2000rhombic}. The proof in fact covers all tilings of the plane by rhombuses, and therefore includes AB tilings. Our own solution to the three-coloring problem on AB tilings takes a different approach, and generates different three-colorings. 

\subsubsection{Application: the Potts model and Protein folding}
\label{subsubsec:3_state_potts}

The $q$-state Potts model, $q\in\mathbb{N}$, is a generalisation of the Ising model to spins $\sigma_i$ which can take one of $q$ values~\cite{wu1982potts}. It is defined by the Hamiltonian
\begin{equation}
H=-J\sum_{\langle ij\rangle}\delta_{\sigma_i,\sigma_j}
\label{eqH}
\end{equation}
where the sum is over nearest neighbors and $\delta$ is the Kronecker delta~\cite{wu1982potts}. We can define the Potts model on the AB tiling by placing one spin on the centre of each face of the tiling, with the four nearest neighbors defined to be the spins situated on the faces reached by crossing one edge. 

The Potts model has a broad range of applications in statistical physics and beyond. It shows first and second order phase transitions, and infinite order BKT transitions, under different conditions of $J$ and $q$. It is used to study the random cluster model~\cite{beffara2012self}, percolation problems~\cite{selke1983interfacial} and the Tutte and chromatic polynomials~\cite{sokal2005multivariate}. Its physical applications include quark confinement~\cite{alford2001solution}, interfaces, grain growth, and foams~\cite{selke1983interfacial}, and morphogenesis in biological systems~\cite{graner1992simulation}. While most extensively studied on periodic lattices, there is some work on the Potts model in Penrose tilings; this model appears to show the universal behaviour present in the periodic cases~\cite{wilson1988evidence, xiong1999real}. 

For $J>0$ the model has a ferromagnetic ground state with all spins aligned along one of their $q$ directions. The behaviour for $J<0$ is more complex. For $q=2$ the model is the Ising model; since a two-coloring of the faces of AB is impossible (ruled out by the existence of 3-vertices), the antiferromagnetic ground state must be geometrically frustrated, meaning the connectivity of sites causes spins to be unable to simultaneously minimise their energies. 

However, the existence of a face-three-coloring proves that for $q\geq 3$ the $J<0$ state is again unfrustrated. Three-colorings give the possible ground states of $q=3$, and a subset of ground states for $q\geq 4$.

The relationship between the Potts model and the process of protein folding has been widely studied in the field of polymer physics and computational biology \cite{levy2017potts,garel1988mean,de2017statistical}. The $q$-state Potts model has been used to model the thermodynamics of protein folding~\cite{potts1952some}. In this scenario the lattice is thought of as a solvent in which the protein exists. A protein is a chain of amino acids \cite{creighton1990protein}. When dissolved it may achieve a stable configuration, typically remarkably compact, which enables its biological functions~\cite{DillEA95}. The folding process depends on the pair-wise interactions between amino acids, set by the interaction with the solvent~\cite{garel1988mean,creighton1990protein,chan1989compact}. If the interactions between similar amino acids are attractive, they cluster together to form a compact folded structure. If these interactions are repulsive, similar amino acids try to spread out over the solvent sites, resulting in a protein which is folded, but not compactly~\cite{de2017statistical}. For modelling proteins, the Hamiltonian in Eq.~\eqref{eqH} represents interacting amino acids $\{\sigma_{i}\}$. For $J<0$ similar amino acids will have repulsive interactions in a solvent according to Eq.~\eqref{eqH}. For $J<0$ and protein chains consisting of only three different types of amino-acids (each corresponding to a color), our three-coloring of AB also predicts one of the stable, minimum energy configurations in which the protein can fold such that each of the three types of amino acid sits on the centre of a tile of corresponding color. 

It seems reasonable to suppose that folding of long protein chains might be better modelled by approximating the solvent with a quasicrystalline surface rather than a simple periodic lattice. This is because quasicrystals provide an aperiodic surface with range of bond angles, unlike regular lattices which have restricted freedom. For instance, square and hexagonal lattices have only one bond angle each ($90^\circ$ and $120^\circ$ respectively), and one vertex type, whereas AB has three bond angles of $45^\circ$, $90^\circ$, and $135^\circ$, and a wide range of different local environments branching out from seven vertex types. The 8-fold rotational symmetry of patches of AB is of a higher degree than is possible in any crystal (six being the maximum degree permitted by crystallographic restriction), additionally making AB a closer approximation to the continuous rotational symmetry of space.

%
\section{Conclusions}
\label{sec:conclusions}
%

We provided exact solutions to a range of problems on aperiodic long-range ordered Ammann-Beenker (AB) tilings, with a further three problems solved in Appendix~\ref{app:appendixB}. There are undoubtedly many more non-trivial problems like these which can be solved on AB tilings by taking advantage of its discrete scale invariance.

To our knowledge only one of the problems we solved was covered by an existing result: the three-color problem was covered by a solution to the three-color problem for Penrose tilings which applied to all tilings of rhombuses~\cite{sibley2000rhombic}. It is reasonable to ask whether this previously known result might similarly have been used to solve a range of other problems on AB, in the same way we exploited solutions to the Hamiltonian cycle problem. While we cannot rule this out, three-colorings do not obviously seem to provide an organising principle for the tiling in the way that our hierarchically constructed loops do.

Our construction relies on the fact that AB tilings are long-range ordered, allowing exact results to be proven in contrast to random structures. It is natural to wonder whether any of our results apply more generally in other quasicrystals. Deleting every second edge along a Hamiltonian cycle creates a perfect dimer matching (division of the vertices into unique neighboring pairs). The existence of such a matching on a graph is therefore a necessary condition for the existence of Hamiltonian cycles. It is not a sufficient condition, as there exist efficient methods, such as the polynomial-time Hopcroft Karp algorithm, for finding perfect matchings~\cite{hopcroft1973n}. Seeking quasicrystals which admit perfect matchings might be a good starting point for finding other quasicrystals which admit Hamiltonian cycles. As a preliminary check, using the Hopcroft Karp algorithm we found the minimum densities of unmatched vertices on large finite patches of quasiperiodic tilings generated using the de Bruijn grid method, with rotational symmetries from 7-fold up to 20-fold~\cite{deBruijn81A}. We found finite densities in each case, suggesting AB may be unique among the de Bruijn tilings in admitting perfect matchings or Hamiltonian cycles\footnote{We used Josh Colclough's code for generating arbitrary de Bruijn grids, available at github.com/joshcol9232/tiling.}. On the other hand, a preliminary check of the recently discovered `Hat'~\cite{smith2023aperiodic} and `Spectre'~\cite{smith2023chiral} aperiodic monotiles (single tiles which tile only aperiodically) suggest that these also admit perfect dimer matchings, even if vertices are added so as to make the graphs bipartite; however, the structure of these matchings seems to indicate that a continuation to fully packed loops is not possible. 

\begin{figure}[h]
    \centering
    \includegraphics[width=\columnwidth]{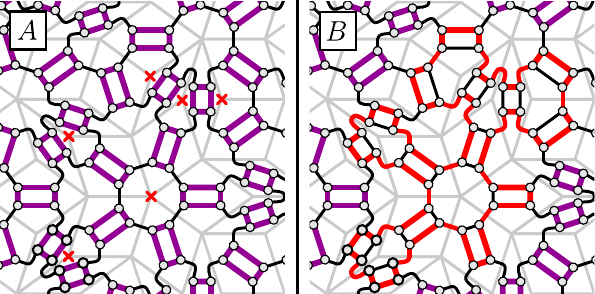}
    \caption{$A$: a Penrose tiling (grey lines) decorated by graph edges (black) and vertices (grey disks) so as to admit fully packed loops (purple) by construction. $B$: augmenting the alternating cycle surrounding each red cross leads to a cycle (red) which visits every graph vertex on any Penrose tile touching a cross.}
    \label{fig:Penrose}
\end{figure}

All of these examples consider the simplest case in which the vertices and edges of the graph are those of the corresponding tiling. Fig.~\ref{fig:Penrose}A shows a Penrose tiling (grey lines) decorated with graph vertices (grey disks) and edges (black) such that the graph admits fully packed loops (purple) by construction. Identical Penrose rhombuses are decorated identically, so that this graph has the full symmetry of the tiling. In Fig.~\ref{fig:Penrose}B we have augmented the alternating cycle surrounding any red cross. The result is a cycle (red) visiting every vertex on any Penrose tile adjoining a cross. Hence, this is a Hamiltonian cycle on an appropriate subgraph. The construction (based on a dimer construction in Ref.~\cite{flicker2020classical}) closely mirrors that for AB tilings: establishing fully packed loops (in this case by construction), before connecting the loops together. Using similar approaches we might reasonably expect our results to generalise to a much broader class of structures.

Our Hamiltonian cycle construction included a proof of the existence of fully packed loops on arbitrarily large finite subgraphs of AB, in which every vertex is visited by a loop, but these need no longer be the same loop~\cite{gier2009fully}. The FPL model is important for understanding the ground states of frustrated magnetic materials such as the spin ices~\cite{jaubert2011analysis}. Famously, loops played a key role in Onsager's exact solution to the Ising model in the square lattice, one of the most important results in statistical physics~\cite{onsager1944crystal, feynman2018statistical}. If all loops could be enumerated systematically, this would suggest the Ising model with spins defined on the faces of AB might be exactly solvable. We note that all possible dimer matchings were exactly enumerated on a modification of AB in which all the 8-vertices were removed, giving the partition function to the dimer problem~\cite{lloyd2021statistical}. The possible loop configurations can similarly be enumerated in that context.

There is a deep connection between loop and dimer models, and quantum electronic tight-binding models defined on the same graphs. For example, Lieb's theorem states that if a bipartite graph has an excess of vertices belonging to one bipartite sublattice, the corresponding tight-binding model must have at least the same number of localised electronic zero modes making up a zero-energy flat band~\cite{Lieb89}. These zero modes are topologically protected in the sense that they survive changes to the hopping strengths provided these strengths remain non-zero. Even if a bipartite graph has an overall balance of sublattice vertices, it can have locally defective regions in which zero modes are similarly localised~\cite{BholaEA22}. This occurs for example in the rhombic P3 Penrose tiling, which has an overall balance between biparite sublattices, but local imbalances lead to a finite density of localised electronic zero modes in its tight-binding and Hubbard models~\cite{KohmotoEA86,koga2017antiferromagnetic,DR20}, which can be exactly calculated using dimers~\cite{flicker2020classical}. Even if the bipartite sublattices balance \emph{locally}, as in the AB tiling~\cite{lloyd2021statistical}, it could still in principle be the case that the graph connectivity localises electrons to certain regions. Recent work has identified such zero modes in the Hubbard model on AB~\cite{Koga20}. Many, but not all, of these appear to fit a loop-based structure first identified in Penrose tilings~\cite{KohmotoEA86} and subsequently more general graphs~\cite{Sutherland86}. The same structures have recently been identified in the Hat aperiodic monotiles~\cite{schirmann2023physical}. A deeper connection between the graph connectivities leading to electronic zero modes, and the loops we identify here, is an interesting avenue for further work.

Hamiltonian cycles and other loop models have been widely used for modelling protein folding, by approximating the molecules as living on the sites of a periodic lattice \cite{bodroza2013enumeration, de2017statistical,chan1989compact}. This simplification has allowed the tools of statistical physics to be brought to bear on complex problems~\cite{chan1989intrachain,chan1989compact,lau1989lattice}. Quasicrystals may well serve as better models for such studies in which space is approximated with a discrete structure. Quasicrystals are able to feature higher degrees of rotational symmetry than crystals, providing a better approximation to the continuous rotational symmetry of real space~\cite{Senechal}. They feature a wider range of nearest neighbor bond angles, and the variety of structures at greater numbers of neighbors increases rapidly. Quasicrystals also lack discrete translational symmetry, a feature of lattices which could cause them to poorly approximate continuous space, although quasicrystals such as AB do feature a discrete scale invariance which might lead to similar shortcomings. Random graphs also have these benefits, but they lack the long-range order of quasicrystals required to obtain exact results such as those we have presented here.

In the context of NPC problems, exact solutions to Hamiltonian cycles are useful because they serve as bench-marking measures for new algorithms \cite{baniasadi2014deterministic, baniasadi2018new, sleegers2022hardest}. While the difficulty of NPC means no known algorithm is capable of deterministically solving all instances in polynomial time, several algorithms display a high success rate on many graphs \cite{helsgaun2000effective, baniasadi2014deterministic, applegate2011traveling}. Finding and classifying Hamiltonian cycle instances which prove `hard' for these algorithms has immediate practical applications and can provide key insight into the full NPC problem. A future goal would be to see how these modern Hamiltonian cycle algorithms perform when applied to AB regions admitting Hamiltonian cycles. Since AB sits at the borderline between regular and random graphs, it contains structures not present in either. Additionally, our Hamiltonian cycle construction is highly nonlocal, leveraging our understanding of the discrete scale symmetry of the tiling. Thus, it may prove a difficult test-case: in that scenario, understanding which structures of AB prove hard for solvers would be a worthwhile challenge. We provided our $U_1$ and $U_2$ graphs to the developers of the state-of-the-art `snakes and ladders heuristic' (SLH) Hamiltonian cycle algorithm \cite{baniasadi2014deterministic}, who ran a preliminary test in this direction\footnote{The SLH algorithm was run on an \emph{Intel Core i7-1165G7} machine with \emph{4} cores, \emph{2.80 GHz} clockspeed, and \emph{32 GB} RAM. We thank Michael Haythorpe~\cite{baniasadi2014deterministic} for running the SLH Hamiltonian cycle algorithm.}. $U_1$ contains 464 vertices and $U_2$ contains 14992 vertices. The SLH algorithm found a Hamiltonian cycle on $U_1$ in modest time ($\sim 1s$) but failed to find cycles on $U_2$ before exceeding the allotted memory. Given the large size of $U_2$ however, this is not conclusive evidence that AB is hard to solve. A next step would be a more thorough analysis, possibly identifying Hamiltonian cycle regions of more reasonable size for bench-marking purposes.

Another interesting direction might be to consider directed or weighted extensions of the graphs. Directed graphs can model non-Hermitian systems with gain and loss~\cite{metz2019spectral}; weighted graphs could model, for example, the travelling salesperson problem in which cities are connected by different distances. Finding a shortest route through a weighted graph is still an NPC problem, and the route would be selected from among the unweighted Hamiltonian cycles. 

Our solution to the longest path problem shows that sets of flexible molecules of arbitrary lengths (not necessarily the same) can pack perfectly onto AB quasicrystals if they register to the atomic positions. This suggests applications to adsorption, which is relevant for instance to carbon capture~\cite{Webley20} and catalysis: two molecules might find a lower-energy route to reacting if they first bond to an appropriate surface before being released after reaction. Physical quasicrystals with three-dimensional point groups -- `icosahedral' quasicrystals (iQCs) -- are known to have properties which make them potentially advantageous as adsorbing substrates. Being brittle, they are readily broken into particles tens of nanometres in diameter, maximising their surface area to volume ratio (a key consideration in industrial catalysis). They are stable at high temperatures, facilitating increased rates of reactions~\cite{LL08}. And, as discussed above, their surface atoms, lacking periodic arrangements, permit a wider range of bond angles. On the other hand, iQCs have been found to have high surface reactivities, which is unfavorable for catalysis because adsorbing molecules will tend to bond to the iQC substrate rather than to one another~\cite{Papadopolos02,Papadopolos04}. They also have low densities of states at the Fermi level compared to metals, which is unfavorable~\cite{FlickerVanWezel15B}. However, iQC nanoparticles have been successfully coated with atomic monolayers which then adopt their quasicrystallinity~\cite{B1,B2}; these monolayers include silver, a widely used and efficient catalyst with none of the quasicrystals' shortcomings~\cite{Pagliaro20}. Investigating the use of quasicrystals for such applications would therefore seem a profitable direction for future experiments. 

The surfaces of iQCs have not been seen to feature the symmetries of AB; rather, they form a class of structures which includes Penrose tilings~\cite{Papadopolos02,Papadopolos04}. Further work is therefore needed to establish the possible packing densities of realistic iQC surfaces. Alternatively, the results presented here could potentially be put to use in the known two-dimensional `axial' quasicrystals which feature the symmetries of AB~\cite{wang1987two,Wang88,Rabson91}.

While AB tilings might have appeared to be an unlikely place to seek exact results to seemingly intractable problems, owing to their lack of periodicity, their long-range order makes such results obtainable. While periodic lattices feature a discrete translational symmetry which simplifies many problems (essentially enabling fields such as condensed matter physics, floquet theory, and lattice QCD), quasicrystal lattices such as AB feature instead a discrete scale invariance. While less familiar, we would argue that this is just as powerful a simplifying factor. For example, the study of the electronic theory of solids is made possible by Bloch's theorem, which allows the electronic wavefunction to be calculated in infinite periodic lattices. An equivalent to Bloch's theorem can be formulated for systems with discrete scale invariance (\emph{e.g.}~Eq.~1.4 in Ref.~\onlinecite{Ovdad21}), suggesting a comparable simplification in seemingly complex problems. To take another example, conformal theories, describing for example critical states at phase transitions, feature both continuous scale and continuous translation invariance. Our work opens the possibility of studying the breakdown of conformal symmetry into discrete scale invariance rather than the more familiar discrete translational invariance~\cite{Gokmen23,Biswas23}. For example, we have established the possibility of studying fully packed loops on AB tilings; in periodic settings these admit conformal field theory descriptions~\cite{kondev1997liouville}. The results presented here provide a key example of the simplifying power of discrete scale invariance.

Broadly, we might expect that many problems for which crystals are special cases might now find quasicrystals to also be special cases, albeit of a fundamentally different form.

%
\section*{Acknowledgments}
%

We thank J.~C.~S\'{e}amus Davis for providing estimates for scanning tunneling microscopy times, Francesco Zamponi for suggestions on Hamiltonian cycle bench-marking, Michael Haythorpe for running the SLH Hamiltonian cycle finder algorithm on AB regions, Irene Giardina and Henna Koivusalo for helpful discussions, Michael Baake, Zohar Ringel, Steven H.~Simon, Michael Sonner and Jasper van Wezel for helpful comments on the manuscript, and Josh Colclough for the use of his code for generating de Bruijn grids. F.~F.~acknowledges support from EPSRC grant EP/X012239/1.

\bibliography{references}

%
\appendix
\counterwithin{figure}{section}
\section{Detailed Figures}
\label{app:appendixA}

In this appendix we include detailed figures to accompany the proof sections in the main text. Fig.~\ref{fig:AP1} and Fig.~\ref{fig:AP2} show intermediate steps for obtaining the Hamiltonian cycle in Fig.~\ref{fig:giant_H_cycle}. In Fig.~\ref{fig:AP1}, black $e_0$ edges form loops visiting every $8_{n<0}$-vertex (i.e.~every vertex which is not an $8$-vertex). We placed red $e_1$ edges along the $e_1$-loops (red loops in Fig.~\ref{fig:thermodynamic_limit}) to add $8_0$-vertices onto the loops. After this we also placed blue $e_2$ edges along the $e_2$-loop (blue loop in Fig.~\ref{fig:thermodynamic_limit}), with one corner folded, to add the central vertex and $8_1$-vertices onto the loops. Note that some parts of the blue $e_2$ edges overlaps with the red $e_1$ edges. Noting that augmentation of a set of edges is reversed by a second augmentation of the same set of edges, in Fig.~\ref{fig:AP2} these overlapping $e_2$ edges and $e_1$ edges will augment to leave the $L_0$ tiling unchanged. Hence we can remove these overlapping edges completely. This results in a rewiring of $e_1$-loops and the $e_2$-loop into a single loop as shown. Augmenting this single loop, consisting of the remaining red and blue edges, will give us the Hamiltonian cycle shown in Fig.~\ref{fig:giant_H_cycle}. 

Fig.~\ref{fig:threecol} is a bigger version of Fig.~\ref{fig:3colSQ}, showing a solution for the 3-coloring problem on a large patch of AB tiling resulting by placing tile 1 and tile 2 shown in Fig.~\ref{fig:3coltiles}, and coloring the remaining ladders consistently on the basis of colors of their surrounding rhombuses such that no two adjacent faces share a color.

%
\onecolumngrid

\begin{figure}[H]
    \centering
    \includegraphics[width=18cm]{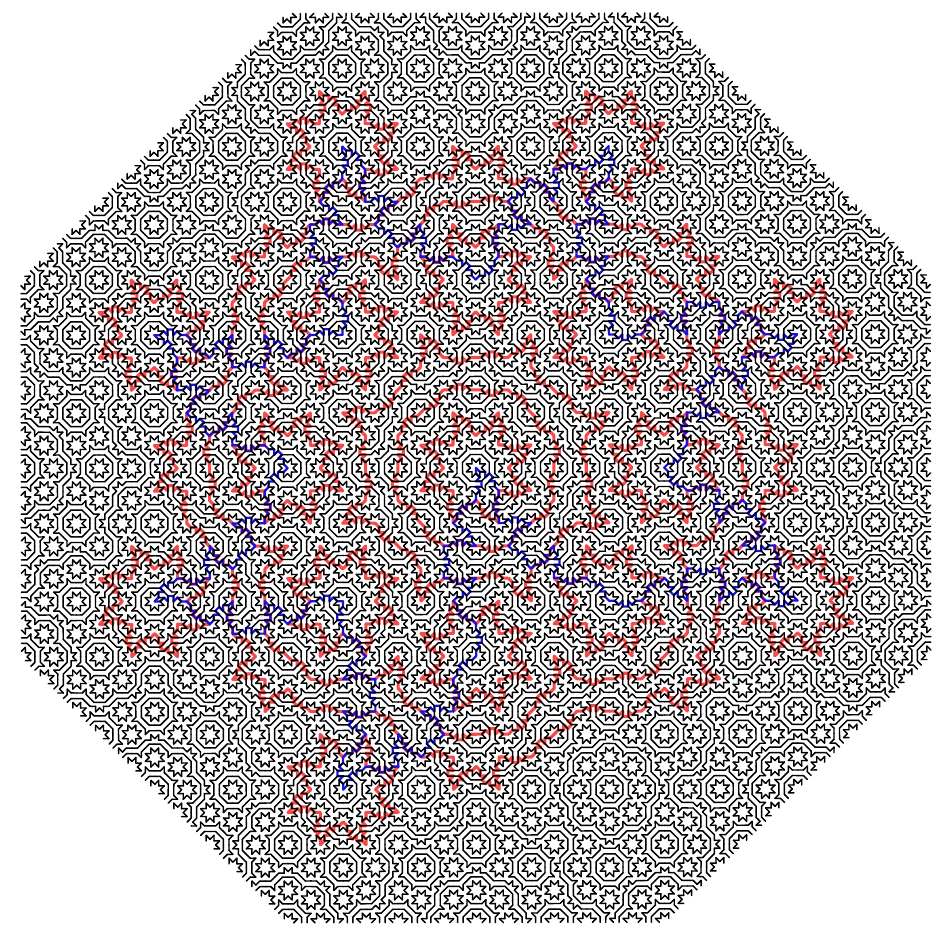}
    \caption{The canonical placement of $e_0$ edges (black) form loops visiting every $8_{n<0}$-vertex (i.e.~every vertex which is not an $8$-vertex). Now $e_1$ edges (red) are placed along the $e_1$-loops (see Fig. \ref{fig:thermodynamic_limit}) to add $8_0$-vertices onto the loops. Further $e_2$ edges (blue) are placed along the $e_2$-loop, with one corner folded, to add the central vertex and $8_1$-vertices onto the loops. Note that some  blue $e_2$ edges overlap with the red $e_1$ edges. This process can be iterated to any order.}
    \label{fig:AP1}
\end{figure}

\begin{figure}[H]
    \centering
    \includegraphics[width=18cm]{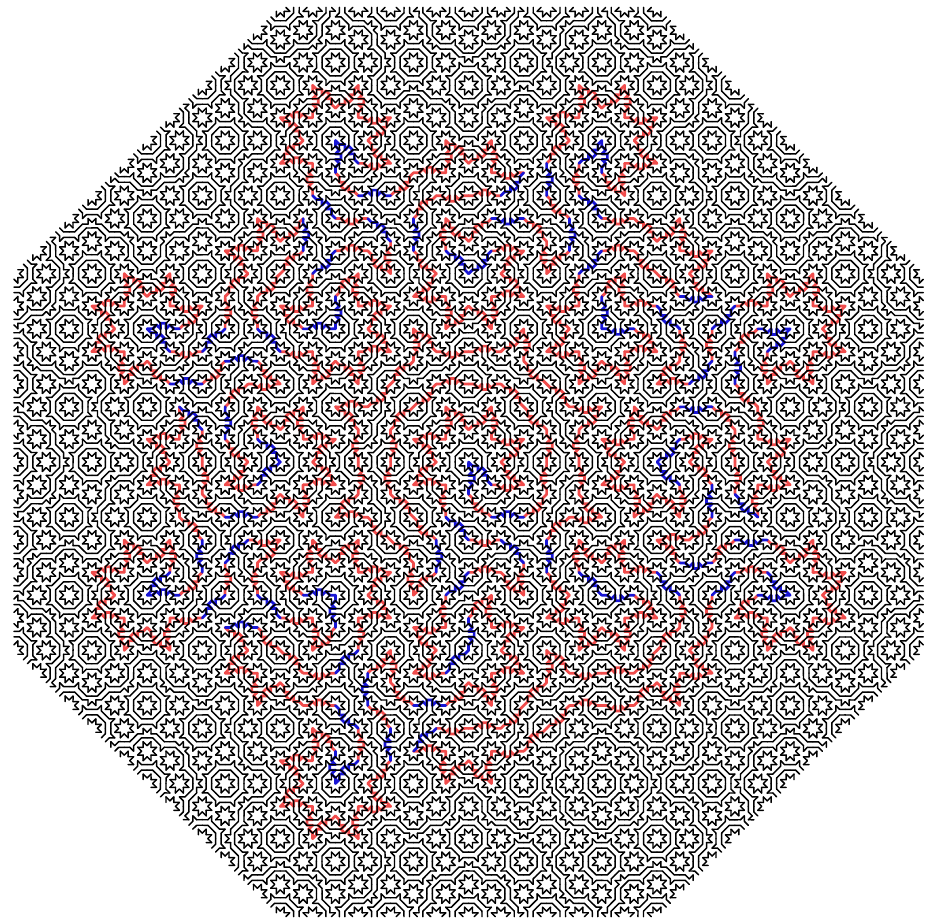}
    \caption{The overlapping $e_2$ edges and $e_1$ edges as shown in Fig~\ref{fig:AP1} will augment to leave the $L_0$ tiling unchanged. Hence we can remove these overlapping $e_2$ edges and $e_1$ edges completely. This results in the rewiring of $e_1$-loops and the $e_2$-loop into a single loop as shown here. Finally augmenting this single loop, consisting of the remaining red and blue edges, will gives the Hamiltonian cycle shown in Fig. \ref{fig:giant_H_cycle}.}
    \label{fig:AP2}
\end{figure}

\begin{figure}[H]
    \centering
    \includegraphics[width=18cm]{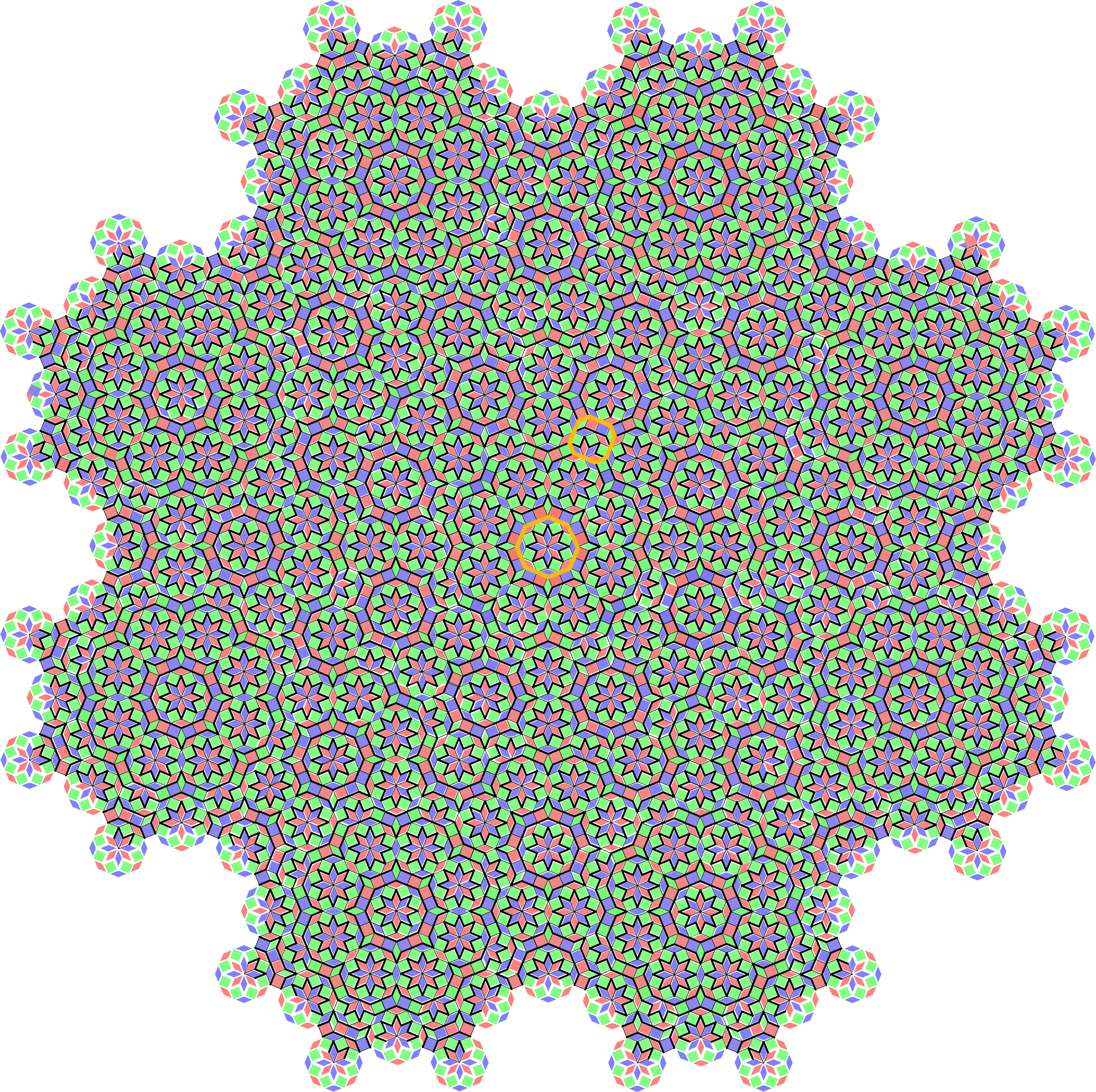}
    \caption{The 3-colored AB tiling using tile 1 and tile 2 (both outlined in gold) as mentioned in Sec. \ref{subsec:impossible6}.}
    \label{fig:threecol}
\end{figure}

\counterwithin{figure}{section}
\section{Solutions to other non-trivial problems on AB}
\label{app:appendixB}
\subsection{The Minimum Dominating Set Problem}
\label{subsec:impossible3}
A dominating set $D$ of a graph $G$ is a subset of the vertices of $G$ such that all the vertices not in $D$ are adjacent to at least one vertex in $D$. The minimum dominating set (MDS) is a dominating set containing the fewest possible vertices; its cardinality is called the domination number, $\gamma (G)$. 

\textbf{Problem statement~\cite{garey1979computers}:} given a graph $G$ and integer $k$, is the domination number $\gamma\left(G\right)\leq k$? The corresponding optimisation problem is to find the MDS.

\textbf{Solution:} if $G$ is built from $N$ copies of $L_{1/2}$ tiles, yes iff $N\leq k$.

\textbf{Proof:} Fig.~\ref{fig:inflation1} shows the single-inflation rules for the rhombus and square, level $L_{1/2}$. The set of red vertices forms an MDS for all vertices within each $L_{1/2}$ tile. Since these inflated tiles cover the entire tiling (owing to discrete scale symmetry), the union of minimum dominating sets of all these inflated tiles must be a dominating set for any patch of the AB tiling, or the tiling itself in the thermodynamic limit. That this is actually an MDS can be seen from the fact that there is no redundancy in the placement of the red vertices in Fig.~\ref{fig:inflation1}: no red vertex has another red as a neighbor, which implies that the set constitutes an MDS for the whole AB tiling. A larger region of the MDS is shown in Fig.~\ref{fig:mds}. Note that the MDS is not unique. 

{\bf Comment}: the proof relies only on the discrete scale symmetry of AB tilings, without reference to the Hamiltonian cycle construction. 

\begin{figure}[h]
\includegraphics[width=8.5cm]{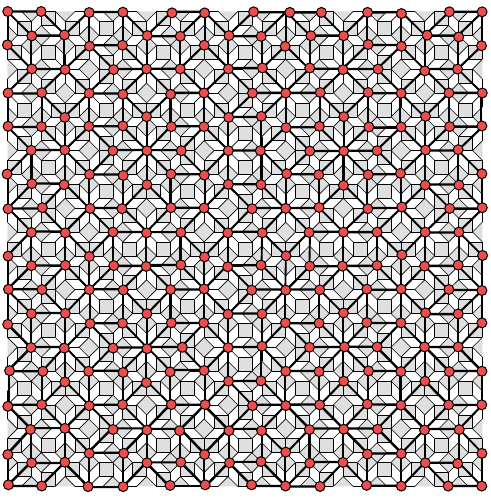}
\caption{A portion of AB tiling showing the minimum dominating set vertices (red), which as proved in Section \ref{subsec:impossible3} are also vertices of single inflated tiles.}
\label{fig:mds}
\end{figure}

\subsection{The Domatic Number Problem}
\label{subsec:impossible4}
The domatic number $d(G)$ for any given graph $G$ with a set of vertices $V$ is the maximum number of disjoint dominating sets into which $V$ can be partitioned.

\textbf{Problem statement~\cite{cockayne1975optimal, garey1979computers}:} for a given graph $G$ and integer $3\leq k \leq \delta\left(G\right)+1$, where $\delta\left(G\right)$ is the minimum degree of $G$, does $G$ have a domatic number $d\left(G\right)\geq k$? The corresponding optimisation problem is to find the corresponding dominating sets.

\textbf{Solution:} if $G=U_{n>0}$, yes if $k = 3$.

\textit{Proof:} we will prove the result by using our Hamiltonian cycles to explicitly construct a partition of $U_{n>0}$ into three disjoint dominating sets.

\begin{figure}[h]
\includegraphics[width=9cm]{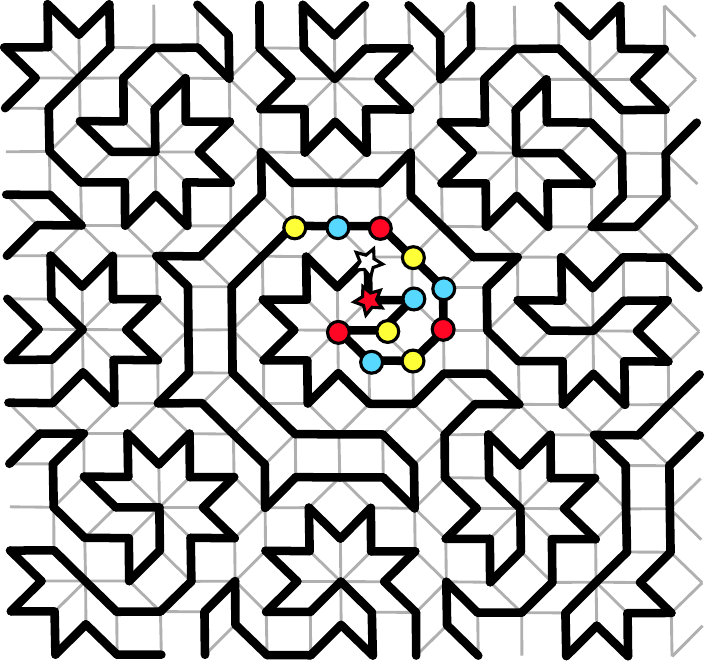}
\caption{Hamiltonian cycle structure (black edges), with vertices colored cyclically by the three colors $R$, $B$, $Y$, starting from the red star and proceeding clockwise around the cycle. The coloring scheme ensures that every vertex of a given color is neighbor to at least one vertex of each of the other colors; the final vertex in the cycle to be colored (white star), neighbors at least one vertex of each color, by construction. This proves that Hamiltonian subgraphs of AB have a domatic number $d\geq 3$.}
\label{fig:d3}
\end{figure}

To see that the domatic number cannot be greater than $\delta+1$, consider the set $N(v)$ composed of a vertex $v$ along with its $\delta$ neighbors, such that $|N(v)| = \delta +1$; then each dominating set $D_i$ contains at least one vertex in $N(v)$, and each vertex in $N(v)$ is contained in at most one $D_i$. Thus $d \leq \delta +1$. The case $k=2$ also trivially holds for any graph with no isolated vertices. Hence, finding $k=3$ disjoint dominating sets is the first non-trivial instance of the problem, which is NPC on general graphs. 

Before proving our statement for AB, we first anticipate the result by giving a simple proof that any Hamiltonian graph $G$ of size $|G| = 3n$, for $n$ integer, admits $k=3$ disjoint dominating sets. We have not seen this statement made elsewhere. To see that it is true, we begin by labelling the three sets by the colors red, blue, and yellow, denoted $R,\ B,\ Y$. Starting from an arbitrary vertex $v_1$ of $G$, color $v_1$ red. Then, proceeding along the Hamiltonian cycle, color each vertex cyclically according to the three colors. Thus, $v_2$ is colored $B$, $v_3$ is colored $Y$, $v_4$ is colored $R$, and so on. The constraint $|V| = 3n$ ensures that the final vertex, $v_{3n}$, neighboring $v_0$, is colored $Y$. In this way the cyclic structure of the coloring holds. Since (\emph{i}) the Hamiltonian cycle visits every vertex and (\emph{ii}) every vertex of a given color neighbors at least one vertex of each of the two remaining colors, we have constructed three disjoint dominating sets for $k=3$. 

The extension to the $G= U_{n>0}$ regions of AB, which we have proven admit Hamiltonian cycles, is similarly simple. The cyclic coloring along the Hamiltonian cycle proceeds exactly as above. However, we do not have $|U_n| = 3n$ (for example, $|U_1|=464|$). Nevertheless, in Fig.~\ref{fig:d3} we show a suitable starting vertex of the Hamiltonian cycle (red star) such that the final vertex of the cycle (white star) neighbors one vertex of each color. The pictured section of the Hamiltonian cycle appears as the central region of the $U_1$ cycle shown in Fig.~\ref{fig:AB_H_cycle}, and every Hamiltonian cycle constructed on $U_{n>0}$ according to our method follows an identical path at its centre. Therefore, every vertex of a given color neighbors at least one vertex of each of the two remaining colors, and we have proven the case $k=3$ for $U_{n>0}$. Note that $U_0$ is too small for this method to hold, and has $d=2$. 

\textbf{Comments:} given that the minimum connectivity of AB is $\delta=3$, the problem statement requests proof that the domatic number of $U_{n>0}$ is either $d=3$ or $d=4$. We have constructed a $d=3$ partitioning using our Hamiltonian cycle construction. This solves the domatic number problem, which is NPC on general graphs. 

The question of whether there exists a $d=4$ partition remains open. We have been able to find such partitions of large regions of AB. We have not found any noticeable decrease in the freedom to continue the patterns as they grow. On the basis of these observations, we conjecture that a $d=4$ partition is possible for arbitrarily large $W_n$ regions. If true, this would render these arbitrarily large finite patches of AB tilings rare examples of `domatically full' graphs~\cite{cockayne1977towards} with $d(G)=\delta(G)+1$.

\subsection{The Induced Path Problem}
\label{subsec:impossible5}

An induced path in an undirected graph $G$ is a sequence of vertices such that a pair of vertices is adjacent in the sequence iff the vertices have an edge in $G$. An induced cycle is an induced path which closes.

\textbf{Problem statement~\cite{Yannakakis,garey1979computers}:} for a graph $G$ and a positive integer $k$, does $G$ contain an induced cycle of length at least $k$? The corresponding optimisation problem is to find an induced cycle of length $k$.

\textbf{Solution:} yes, if $G=$~AB.

\textbf{Proof:} in this case it is convenient to consider the $D_8$-symmetric regions $W_n$. Fig.~\ref{fig:IP} shows the loops of the FPL on AB* constructed in Sec.~\ref{subsec:FPL_AB*}. We see that any non adjacent vertices on these loops do not have an edge in $G$, while any adjacent vertices do. Hence these loops are induced cycles. The longest such loop in $W_n$ is of length $8\left(9^n+1\right)$ (this can be proven using the inflation rules of Fig.~\ref{fig:inflation1}). Since any AB tiling contains $W_\infty$, the answer to the problem is `yes' for any $k$, finite or infinite. 

\textbf{Comments:} on general graphs the problem remains NPC if an induced path is sought rather than an induced cycle. That problem can also be solved on AB, by deleting any pair of consecutive edges along the cycle. Longer induced cycles than those we consider can be found in $W_n$, but the problem statement does not require them. 

\begin{figure}[h]
\includegraphics[width=9cm]{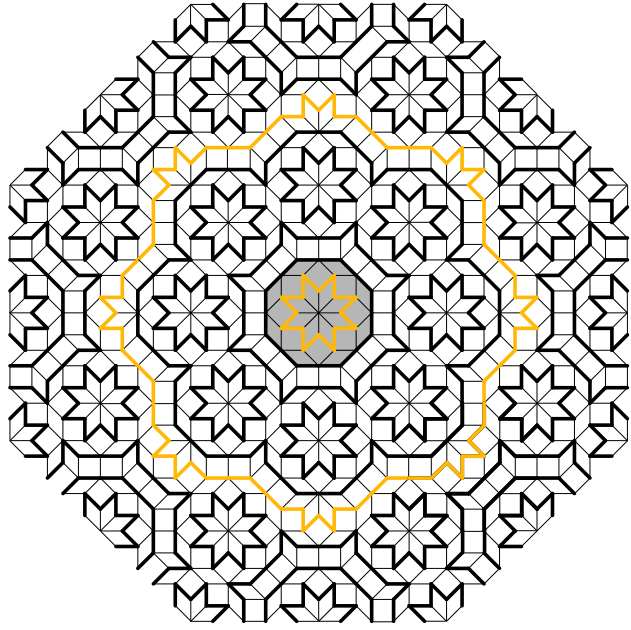}
\caption{FPLs on AB* formed by the method in Section~\ref{subsec:FPL_AB*} form induced paths on AB. The longest induced path in region $W_n$ created by this method is of length $8\left(9^n+1\right)$. The total region shown here is $W_1$, with $W_0$ highlighted in grey. The corresponding induced paths are shown in gold.}
\label{fig:IP}
\end{figure}

\end{document}